\documentclass[10pt]{article}
\usepackage{setspace}
\usepackage{color}
\usepackage{lipsum}
\usepackage{verbatim}
\usepackage{amsfonts,amsmath,amssymb,mathrsfs,amsthm}
\usepackage{graphicx}
\usepackage{epstopdf}
\usepackage{natbib}
\usepackage{algorithmic}
\usepackage{amsopn}
\usepackage{mathtools}
\usepackage{cases}
\usepackage{enumerate}
\usepackage{graphicx,subfigure}
\newtheorem{theorem}{Theorem}[section]
\newtheorem{lemma}{Lemma}[section]
\newtheorem{definition}{Definition}[section]

\newtheorem{proposition}{Proposition}[section]
\newtheorem{corollary}{Corollary}[section]
\newtheorem{example}{Example}[section]
\usepackage{geometry}
\renewcommand{\epsilon}{\varepsilon}
\newcommand{\esssup}{\mathop{\operatorname{ess\,sup}}}

\newtheorem{thm}{Theorem}[section]

\newtheorem{rem}[thm]{Remark}

\numberwithin {equation} {section}
\def\be{\begin{equation}} 
\def\ee{\end{equation}} 

\begin{document}
\title{\LARGE \bf Optimal Consumption-Investment with Epstein-Zin Utility under Leverage Constraint}
\author{ Dejian Tian\thanks{School of Mathematics, China University of Mining and Technology, China. Email: djtian@cumt.edu.cn.} \qquad
Weidong Tian\thanks{Belk College of Business,
University of North Carolina at Charlotte, USA.  Email: wtian1@charlotte.edu.}
\qquad
Jianjun Zhou \thanks{Corresponding author, College of Science, Northwest A\&F University, China. Email: zhoujianjun@nwsuaf.edu.cn.} \qquad
Zimu Zhu\thanks{ Fintech Thrust, Hong Kong University of Science and Technology (Guangzhou), China. Email: zimuzhu@hkust-gz.edu.cn.}
}

\date{}
\maketitle
\thispagestyle {empty}

\newpage
\renewcommand{\abstractname}{
{\LARGE Optimal Consumption-Investment with Epstein-Zin Utility under Leverage Constraint }
\\[1in] Abstract}

\begin{abstract}
We study optimal portfolio choice under Epstein–Zin recursive utility in the presence of general leverage constraints. We first establish that the optimal value function is the unique viscosity solution to the associated Hamilton–Jacobi–Bellman (HJB) equation, by developing  a new dynamic programming principle under constraints. We further demonstrate that the value function admits smoothness and characterize the optimal consumption and investment strategies. In addition, we derive explicit solutions for the optimal strategy and explicitly delineate the constrained and unconstrained regions in several special cases of the leverage constraint. Finally, we conduct a comparative analysis, highlighting the differences relative to the classical time-separable preferences and to the setting without leverage constraints.
\vspace{0.9cm}

\noindent\textit{Keywords}:  Epstein-Zin utility, Leverage constraint, Viscosity and smooth solution, Dynamic programming  principle  \\
\noindent\textit{Mathematics Subject Classification (2020)}: 49L20, 60H20, 91G10, 
91G80, 93E20 \\
\noindent\textit{JEL Classification Codes}: C61, G11
\end{abstract}
\thispagestyle{empty}

\newpage

\setcounter{page}{1}
\section{Introduction}

In this paper, we study a new class of optimal portfolio choice problems under constraints in a continuous-time setting, motivated by a central strand of the asset pricing literature. Specifically, an investor has an Epstein-Zin preference with aggregator
\begin{align*}
f(c,v)=\frac{c^{1-S}}{1-S}\left((1-R)v\right)^{\rho},
\end{align*}
where $R > 0$ denotes the agent's relative aversion and $S>0$ captures the agent's elasticity of intertemporal complementarity (EIC). 
Here $\rho = \frac{S-R}{1-S}, R \ne 1, S \ne 1$. In a standard Black–Scholes financial market, the corresponding optimal portfolio choice problem {\em without any constraints} has been extensively studied in \cite{HHJ23a}, \cite{HHJ23b}, \cite{HHJ25} and \cite{S25}. We extend this literature by studying the optimal portfolio choice problem under a {\em leverage constraint} on the risky position, where the dollar investment in the risky asset is controlled by an exogenous general function of wealth at any time.

This paper offers the first theoretical analysis of the problem, except in the special case $R = S$, which corresponds to time-separable (CRRA) preferences. In that setting, \cite{VZ97} and \cite{Z94} obtain results that are closely related to those presented here. The main challenge in this paper is that the aggregator $f(c,v)$ is not Lipschitz continuous in either component. Moreover, the leverage constraint is time-varying and the set of feasible strategies is  not compact. Consequently, the arguments for the time-separable preference, which rely primarily on standard stochastic control and viscosity solution theory developed in works such as \cite{c92}, \cite{DN90}, and \cite{FS06}, cannot be applied directly in this setting. This paper develops a new analytical approach that integrates backward stochastic differential equations (BSDEs), Epstein–Zin utility, and stochastic control with time-varying and state-dependent constraints.

Specifically, this paper makes several major theoretical contributions to the literature. In the first main result, Theorem \ref{th-unique}, we establish that the optimal value function is the unique viscosity solution of the corresponding HJB equation under appropriate boundary conditions. Our novel approach is to prove the new type of dynamic programming principle (DPP). By exploiting the relationship between the Epstein–Zin utility process and the infinite-horizon BSDE, we show that the optimal value function satisfies the DPP when the consumption process is truncated at any fixed positive level. This result is a crucial step toward demonstrating that the optimal value function is indeed a viscosity solution of the HJB equation. The proof of uniqueness is also subtle and departs significantly from the classical case. In particular, establishing the comparison principle requires carefully handling both components of the Epstein–Zin aggregator.

The second main result, Theorem \ref{th-c2}, establishes that the optimal value function is a $C^2$ smooth function of the wealth. The proof relies heavily on the aforementioned new viscosity solution theorem and its construction. Additionally, our argument is highly technical, making full use of the specific structure of the Epstein–Zin aggregator and analyzing the second-order properties of a sequence of auxiliary value functions. Given the smoothness of the optimal value function, we are then able to explicitly characterize the optimal consumption and investment strategy.

Our third main contribution is to examine in detail a standard leverage constraint, in which the risky investment is bounded by a linear function of wealth. Although there is no closed-form expression for the optimal value function in the general linear case, two particular cases can be solved analytically in Theorem \ref{th-linear}. First, when the linear function reduces to a proportion of wealth, the constrained region is 
 $(0, \infty)$, and the optimal investment strategy follows the prescribed proportional leverage constraint. Second, when the linear function reduces to a positive constant, the constrained region is $(x^*, \infty)$, where $x^* >0$ is uniquely determined by a smooth-fit condition at this threshold. The latter result is highly nontrivial, as it relies on applying the comparison principle for elliptic differential operators, whereas the standard arguments used in the time-separable preference case fail. In general, we also present several comparative analyses of the optimal consumption and investment strategies.



There have been several important studies on the optimal portfolio choice problem for Epstein-Zin utility, see, for instance, \cite{AH21}, \cite{KSS17}, \cite{KSS13} , \cite{MX18}, \cite{SS99},  \cite{SS16}, \cite{WWY16}, and \cite{X17}. However, these studies, along with the aforementioned papers, do not consider constraints. Moreover, they either adopt a martingale approach or impose an {\em ex ante} smoothness assumption on the value function, subsequently verifying that a proposed analytical expression coincides with the value function in certain special cases. By contrast, this paper establishes the general existence and uniqueness of the viscosity solution, and subsequently proves the smoothness of the value function.

Portfolio constraints, in particular the leverage constraint considered in this paper, have been widely studied in both the probability and
mathematical finance literature because of their economic importance. See, for instance, \cite{C05}, \cite{CL00}, \cite{CK92}, \cite{D95}, \cite{DL11},  \cite{ET08}, \cite{GV92}, \cite{HI22}, \cite{JK24}, \cite{LZ13}, \cite{SS15}, \cite{TZ22}, and \cite{VZ97}. The leverage constraint is also referred to as a borrowing constraint, margin requirement, or liquidity constraint in different contexts. Nevertheless, research on portfolio constraints has focused primarily on time-separable preferences, with only a limited number of studies in the Epstein–Zin framework, including \cite{HHT24}, \cite{MMS20}, and \cite{TTZ25}. The first two investigate transaction costs, while the latter considers a consumption constraint. To the best of our knowledge, this paper is the first to provide a comprehensive theoretical solution to the problem under a general leverage constraint by introducing new analytical methods.

In the economic and finance literature, \cite{EZ89} pioneers the concept of recursive utilities in a discrete-time framework. Since their inception, these Epstein-Zin type preferences have provided a robust framework for addressing numerous asset pricing anomalies, as in the long-run risk finance literature (see, for instance, \cite{BCZ2014}, \cite{BY04},  \cite{BS20}, and \cite{SSY18}). 
While a continuous-time version of Epstein–Zin preferences was introduced in \cite{DE92} under the guise of stochastic recursive preferences, the non-Lipschitz nature of Epstein–Zin utility makes the mathematical theory challenging in continuous time. Consequently, most financial applications have been developed within a discrete-time framework. For instance, \cite{BS20}, \cite{B20}, \cite{GP15} and \cite{MM10}. Moreover, prior studies, such as \cite{BJ18}, \cite{S22}, and \cite{SWZ24}, have demonstrated the dynamic programming principle in the discrete-time setting. Consequently, the analytical approach developed in this paper may provide valuable insights for {\em continuous-time asset pricing} under Epstein–Zin preferences.

The structure of the paper is as follows. Section \ref{sec:model} introduces the problem. Section \ref{sec:basic pro} presents three main theorems of this paper.  Section \ref{sec:dpp} develops a dynamic programming principle tailored to our setting. Section \ref{sec:viscosity} establishes the first main theorem, while  Section \ref{sec:smoothness} proves the second. Section \ref{sec:linear} provides a detailed analysis of the linear leverage constraint and establishes the third main theorem. Section \ref{sec:conclusion} concludes. Some technical proofs are collected in Appendices A–D.


\section{Model setup and the optimization problem}\label{sec:model}

 The probability space is defined as $\left(\Omega,\left(\mathcal{F}_t\right)_{t\geq0}, \mathcal{F}, \mathbb{P}\right)$, where the information flow $\left(\mathcal{F}_t\right)_{t\geq0}$ is generated by a one-dimensional standard Brownian motion $\left(B_t\right)_{t\geq0}$.   $\mathcal{F}$ is the sigma-algebra generated by $\left(\mathcal{F}_t\right)_{t\geq0}$.  We denote $\mathbb{R}_+=[0,+\infty)$, $\mathbb{R}_{++}=(0,+\infty)$ and
$\bar{\mathbb{R}}_+=[0,+\infty]$, respectively.

\subsection{The financial market}
We consider a continuous market with two assets. The first is a risky asset, representing the stock index in the equity market, whose price process follows the dynamics
\begin{equation}
d S_t=\mu S_t d t+\sigma S_t d B_t,  ~~t\geq0,   ~~S_0=s_0>0,
\end{equation}
where $\mu > 0$ and $\sigma > 0$ are constants. The second is a risk-free asset with a constant rate of return $r$, and $0 < r < \mu$. It is natural to use the risk-free asset to represent a bond market in cases where interest rate risk is not a major concern, or when the volatility of the interest rate is dominated by the equity market risk.

The total wealth process evolves according to the equation
\begin{equation}
\label{eq:wealth}
d X_t=r X_t d t+\pi_t (\mu-r) d t+\pi_t \sigma d B_t-c_t d t, ~~t\geq0, ~~~X_0 = x\geq0,
\end{equation}
where $x$ is the initial wealth, $\pi_t$ is the amount of wealth invested in stock and $c_t$ is the consumption rate at time $t$.

Let ${\mathscr P}$ be the set of progressively measurable processes and ${\mathscr P}_{+}$ the restriction of ${\mathscr P}$ to processes that take nonnegative values, respectively.  We introduce the following sets:
\begin{align*}
   \mathcal{L}_+&=\left\{u : u\in\mathscr{P}_+ \text{~~and~} \int_0^tu_sds<+\infty ~a.s.,   \forall t\geq0 \right\},\\
   \mathcal{M}&=\left\{u : u\in\mathscr{P} \text{~~and~} \int_0^tu^2_sds<+\infty ~a.s.,   \forall t\geq0 \right\}.
\end{align*}

For any given $x\geq0$, let $\mathcal{A}(x)$ denote the set of \textit{admissible} consumption-investment strategies $\left(\pi_t, c_t\right)_{t\geq0}$ such that: \begin{itemize}
    \item[(i)] Integrable condition: $c\in\mathcal{L}_+$ and $\pi\in\mathcal{M}$;
    \item[(ii)] Non-negative condition: the wealth process $X_t\geq0$, a.s.,  for all $t\geq0$ under the given strategy $(\pi,c)$;
    \item[(iii)] Time-varying leverage constraints: $|\pi_t|\leq g(X_t)$, a.s., for all $t\geq0$, where the function $g:\mathbb{R}_{+}\rightarrow\mathbb{R}_{+}$ is an increasing, concave function such that $g(0)\geq 0$, $g(x)>0$, for all $x>0$,  and satisfying Lipschitz continuous condition with constant $K\geq0$.
\end{itemize}

\begin{rem}
   For any $x\geq0$, since $(0,rx)\in\mathcal{A}(x)$, the set  $\mathcal{A}(x)$ is not empty. We define a consumption stream $c\in \mathcal{C}(x)$ if there is an investment process $\pi$ such that $(\pi,c)\in \mathcal{A}(x)$.
\end{rem}

In the leverage constraint, there are two components. First, there is an upper bound on the risky investment relative to total wealth. This bound is often described using several  different terminologies.  For instance, when $\pi_t \le X_t$, it implies a long position in the risk-free asset since the dollar investment in the risky free asset, $X_t - \pi_t$ is nonnegative. When $\pi_t > X_t$, it is often referred to as a leverage or margin requirement constraint, since the investor cannot take on excessive leverage in the risky asset due to the condition $\pi_t \le g(X_t)$. For example, the constraint $\pi_t \le 2X_t$ means that the investor does not borrow more than their current wealth to invest in the risky asset. Because borrowing is required to establish such a leveraged position, this is also called a borrowing constraint (see \cite{C05}, \cite{CL00}, \cite{GV92},  and \cite{VZ97}). Moreover, in the corporate finance framework, when the risky asset is interpreted as equity and wealth as total assets, the leverage constraint is imposed as an upper bound on the equity-to-asset ratio (see, for instance, \cite{ACL19}). In this setting, it is sometimes referred to as a liquidity constraint. 

Second, the lower bound on risky investment  is often related to a lower bound on total wealth and ruled out the arbitrage opportunity. For instance, when $\pi_t \ge 0$, it rules out short positions in the risky asset, which corresponds to the standard short-sale constraint. Similarly, a condition such as $\pi_t \ge 0.5 X_t$ requires the investor to allocate at least 50 percent of wealth to the risky asset. Equivalently, this implies  $X_t - \pi_t \le 0.5 X_t$ in the risk-free asset. In our paper, we allow for different specifications of upper and lower bounds, such as $\pi_t \le g(X_t)$ and $\pi_t \ge h(X_t)$. The essential point for the lower bound is to ensure the optimal value function exists and possesses desirable properties. As we demonstrate in the main theorems below, at least within the Black–Scholes framework, the optimal risky allocation is always non-negative under these conditions.

\begin{rem}
     \cite{SS15} considers a compact set restriction on the risky investment, specifically $\pi_t \in [a, b]$ with $0 < a < b$ in a single–risky-asset economy. Their arguments cannot accommodate short-selling, since the dollar amount can be arbitrarily close to zero.
\end{rem}

\begin{example}
    In most applications, the function  $g(x)$ is specified in linear form, and we will discuss this case in detail in Section \ref{sec:linear}. To be specific, when $g(x) = kx + L$, the constraint implies that the risky investment cannot exceed a multiple $k$ of total wealth, with $L$ representing a fixed allowance.
\end{example}

\begin{example}
  We allow for a broad specification of the function $g(x)$. For example, the leverage may depend on several wealth thresholds. Let $W_0 =0 < W_1 < \cdots < W_{N} < W_{N+1} = +\infty$, and $g_i(x) = k_i x + L$ over the region $(W_i, W_{i+1})$. We assume that $k_i$ is decreasing, so that proportional leverage declines as wealth increases. This type of leverage specification can also be addressed by extending the discussion in Section \ref{sec:linear} in a straightforward manner.
\end{example}

\subsection{The Epstein-Zin utility}

Let $R$ and $S$ both lie in $(0,1) \cup (1,+\infty)$, and set $\mathbb{V}=(1-R)\mathbb{\bar{R}}_{+}$.  Let  $\nu=\frac{1-R}{1-S}$ and $\rho=\frac{S-R}{1-R}=1-\frac{1}{\nu}$. We define an aggregator $f:\mathbb{R}_+\times \mathbb{V}\rightarrow \mathbb{V}$  as follows:
\begin{equation}
\label{eq:aggregator}
f(c,v)=\frac{c^{1-S}}{1-S}\left((1-R)v\right)^{\rho},
\end{equation}
where $R$ denotes the agent’s relative risk aversion, and $S = \frac{1}{\psi}$, and the number $\psi$ represents the agent’s elasticity of intertemporal 
substitution (EIS). Since $S$ is the reciprocal of the elasticity of intertemporal substitution, it is often referred to as the elasticity of intertemporal complementarity (EIC). In the  particular case where $R = S$, i.e., $R \psi = 1$, we have $\nu = 1$ and $\rho = 0$, and the utility function reduces to the standard CRRA form.

\begin{rem}
Given the specification of the aggregator in \eqref{eq:aggregator},  it is possible that $\frac{c^{1-S}}{1-S}$, $\left((1-R)v\right)^{\rho}$ take values in $\{0,+\infty\}$. In these cases, we adopt the standard conventions from Section 4 of \cite{HHJ23b}. For example, when $\rho<0$ and $0<R<1$, we take $f(0,0) = 0$.
\end{rem}

Let $\delta > 0$ denote the subjective discount rate, reflecting the agent's rate of time preference.
The stochastic differential utility process $(V_t^{c})$ associates to the consumption process $c \in {\mathscr P}_{+}$ and the aggregator $f$ is one that satisfies:
\begin{equation}\label{eq:ez-utility}
V_{t}^{c}=\mathbb{E}\left[\int_{t}^{\infty}e^{-\delta s} f(c_{s},V_{s}^{c})ds~\big|~\mathcal{F}_{t}\right],~~t\geq 0.
\end{equation}

The aggregator (\ref{eq:aggregator}) differs slightly from the classical (minus) version in the Epstein-Zin literature; however, the corresponding differential utility processes retain all essential properties, as demonstrated in \cite{HHJ23a}. Moreover, the space $\mathscr{P}_{+}$ of consumption streams under the aggregator (\ref{eq:aggregator}) is broader than that of the classical aggregator, eliminating restrictive integral conditions (see, e.g., \cite{SS99}) that limit its applicability.
 In addition, \cite{HHJ23a, HHJ23b} demonstrate that the coefficients $R$ and $S$ must lie on the same side of unity (that is, $\nu > 0$) to ensure a well-defined utility process and prevent bubble formation.

 The existence and uniqueness of the stochastic recursive utility with the Epstein–Zin aggregator (\ref{eq:aggregator}) are far from obvious and are typically linked to a fixed-point problem in both discrete and continuous frameworks. \cite[Theorem 6]{HHJ23b} shows that for any consumption $c\in\mathscr{P}_{+}$, there exists a unique stochastic recursive utility process $V^{c}= \{V_{t}^{c}\}_{t\geq0}$ that solves \eqref{eq:ez-utility} under the assumption $\nu\in(0,1)$. For this reason, we adopt the Epstein–Zin utility formulation described in \cite{HHJ23a} and assume that $\nu\in(0,1)$ in this paper.
 
\begin{rem}
It is worth noting that \cite{S25} provides an economic interpretation for the case $\nu < 0$, in slightly different set of feasible portfolio sets. In contrast, \cite{HHJ25} examines the case $\nu > 1$ by considering an appropriate Epstein–Zin utility process. In this paper, we focus on the case $\nu \in (0,1)$ to remain consistent with previous studies on leverage constraints under time-separable preferences. 
\end{rem}

\begin{rem}
 The condition $\nu \in (0, 1)$ is equivalent to $\rho < 0$. Technically, the classical case $\rho=0$ is excluded in our setting, but this case can be either approximated by a sequence of negative values, such as $\rho_n = - \frac{1}{n}$, or verified directly as in  \cite{Z94}. The key challenge in our setting is that the aggregator involves two variables $(R,S)$, whereas in the classical case it depends on only one variable.
\end{rem}

\subsection{The optimal consumption-investment problem}
The optimal portfolio choice problem is as follows:
\begin{align}\label{eq:ez-problem}
J(x)=\sup _{(\pi, c) \in \mathcal{A}(x)} V_{0}^{c}=\sup _{c \in \mathcal{C}(x)} V_{0}^{c},
~~~~x>0,\end{align}where $V^{c}= \{V_{t}^{c}\}_{t\geq0}$ denotes an appropriate solution of \eqref{eq:ez-utility}.
We take $g\equiv+\infty$ as a benchmark model, which corresponds to the case without borrowing restrictions on the portfolio strategy.  The corresponding optimal portfolio choice problem is as follows:
\begin{align}\label{eq:ez-problem-no-constraints}
J^{ez}(x)=\sup _{(\pi, c) \in \mathcal{A}_{ez}(x)} V_{0}^{c},~~~~~x>0,
\end{align}where $\mathcal{A}_{ez}(x)$ corresponds to the feasible strategies set without borrowing constraints.  Theorem 8.1 in \cite{HHJ23b} establishes that
the optimal strategies are given by:
\begin{align}\label{eq:ez-without constraints}
\pi^{ez}=\frac{\mu-r}{R\sigma^{2}}X  \text{~~~~~and~~~~~}  c^{ez}=\eta X,\end{align} and the optimal value function in this situation is 
\begin{equation}
\label{eq:benchmark}
J^{ez}(x)=\eta^{-\nu S}\frac{x^{1-R}}{1-R},
\end{equation}
under assumption that
\begin{align}\label{eq:eta}
\eta \equiv \frac{1}{S}\left[\delta+(S-1)(r+\frac{\kappa}{R})\right] > 0,
\end{align}
where
\begin{align}\label{eq:kappa}
\kappa=\frac{(\mu-r)^{2}}{2\sigma^{2}}.
\end{align}

For technical reasons, we assume $0<R<1$ throughout the paper, although we note that several results extend to the more general case.

\section{Major Theorems}\label{sec:basic pro}
We present the main theorems of this paper as follows. The proof of each major theorem is given in subsequent sections.  

%
%

\begin{theorem}\label{th-unique}
The value function $J$ is the unique viscosity solution of the following HJB equation
\begin{align}
\label{eq:HJB}
-\delta\nu J(x)+\mathbf{H}(x,J(x),J_x(x),J_{xx}(x))=0, ~ \ x\in \mathbb{R}_+,\ \ 
\end{align}
where
\begin{eqnarray*}
\mathbf{H}(x,k,p,q)&=\sup_{{|\pi|\leq g(x)}}[\pi(\mu-r)p+\frac{1}{2}\sigma^2\pi^2q]+\sup_{c\geq 0}[f(c,k)-cp]+rxp,\\ 
&(x,k,p,q)\in \mathbb{R}_+\times \mathbb{R}_+\times \mathbb{R}\times \mathbb{R}.
\end{eqnarray*}
\end{theorem}

\begin{theorem}\label{th-c2}
The optimal value function $J$ is the $C^2(\mathbb{R}_{++})$ solution of the equation (\ref{eq:HJB}). Moreover, the optimal strategy for \eqref{eq:ez-problem} is given in the feedback form as follows:
\begin{align}
    c^{*}(X)&= (J_x(X))^{-1/S} \big((1-R)J(X)\big)^{\frac{\rho}{S}},\\
    \pi^*(X)&=\min\left\{ -\frac{\mu-r}{\sigma^{2}} \frac{J_x(X)}{J_{xx}(X)}, g(X)\right\},
\end{align}where $X$ is the optimal wealth trajectory determined by \eqref{eq:wealth} with the above $c^*$ and $\pi^*$.
\end{theorem}

Since the optimal value function is  $C^2$ smooth in Theorem \ref{th-c2}, we define the unconstrained domain $\mathcal{U}$ and constrained domain $\mathcal{B}$ as
\begin{equation*}
\mathcal{U}=\left\{x: \Big{|}-{\mu-r\over \sigma^2}{J_x(x)\over J_{xx}(x)}\Big{|} < g(x) \right\},
\end{equation*}
and
\begin{equation*}
\mathcal{B}=\left\{x: \Big{|}-{\mu-r\over \sigma^2}{J_x(x)\over J_{xx}(x)}\Big{|}> g(x) \right\}.
\end{equation*}
Since $J_x>0$ and $J_{xx}<0$,  it is clear that $\Big{|}-{\mu-r\over \sigma^2}{J_x(x)\over J_{xx}(x)}\Big{|}=-{\mu-r\over \sigma^2}{J_x(x)\over J_{xx}(x)}$.

In each constrained and unconstrained region, the optimal value function satisfies the following highly nonlinear ODE as follows: 
\begin{align}
\label{eq:J-U1}
\delta\nu {J} =-\kappa \frac{ (J_{x})^2}{{J_{xx}}}+ {S\over 1-S}((1-R)J)^{\frac{\rho}{S}} (J_x)^{1-\frac{1}{S}}+rx J_x, ~~~x\in \mathcal{U}
\end{align}
and
\begin{align}
\label{eq:J-B}
\delta\nu {J} =& (\mu-r) g(x) J_x + \frac{1}{2} \sigma^2 g(x)^2 J_{xx}+{S\over 1-S}((1-R)J)^{\frac{\rho}{S}} (J_x)^{1-\frac{1}{S}}\\
&+rx J_x , ~~~x\in \cal{B}. \nonumber
\end{align}

We provide a detailed characterization of the solution under a linear leverage bound, specified as $g(x) = kx + L$, where at least one of  $k$ and $L$ is positive. In particular, we obtain the following result.

\begin{theorem}
    \label{th-linear}
    Assume $kL = 0$, the constrained region is an infinite open interval. Specifically, if $L = 0$, we have  
    ${\cal B} = (0, +\infty)$ when $k<\frac{\mu-r}{R\sigma^2}$, and  ${\cal U} = (0, +\infty)$ when $k\geq\frac{\mu-r}{R\sigma^2}$. If $k =0$, then ${\cal B} = (x^*, +\infty)$ where  $x^* > 0$ is finite and uniquely determined by the smooth-fit condition.
\end{theorem}


\section{Dynamic programming principle}
\label{sec:dpp}

The goal of this section is to prove a new type of dynamic program principle for Epstein-Zin utility, which is essential to show Theorem \ref{th-unique} in the next section. The challenge in proving the dynamic program principle is the non-Lipschitz feature of the Epstein-Zin aggregator on both components. We develop a new approach to resolve this issue. 

\subsection{Basic properties of the value function}
We begin with several basic propositions concerning the value function. These properties will be used in the subsequent discussion.

\begin{proposition}\label{pro:basic pro}
The value function $J$, defined in \eqref{eq:ez-problem},  is increasing and satisfies the following growth property:
       \begin{align}
    \label{eq:ez-bound}
    \delta^{-\nu} \frac{(rx)^{1-R}}{1-R} \le J(x) \le \eta^{-\nu S}\frac{x^{1-R}}{1-R},~~~~x>0.
\end{align}
Moreover, $(1-R)J(x)$ is positive for $x > 0$. In particular, when $R\in(0,1)$, then $J(0) = 0$ and  
$J_x(0)=+\infty$.
\end{proposition}

\begin{proof} For any $x_{2}\geq x_{1}>0$, it holds that $\mathcal{C}(x_{1})\subset\mathcal{C}(x_{2})$, which implies $J(x_{1})\leq J(x_{2})$. 

For any $x>0$, 
we consider a feasible strategy $(\pi, c) = (0, rx)$ and the corresponding stochastic differential utility process $V_t^{(rx)}$. In this case, the stochastic differential utility satisfies the ordinary differential equation: $$d\bar{V}_t=-\nu a^{1-S}e^{-\delta t}\bar{V}_t^{\rho}dt,~~~\bar{V}_\infty=0$$ with $\bar{V}_t=(1-R)V_{t}$. Solving this equation yields: 
\begin{align}
\label{eq:lower}
    V_t^{(rx)} = \delta^{-\nu} e^{-\delta \nu t} \frac{(rx)^{1-R}}{1-R},
\end{align}
which serves as the lower bound for $J(x)$ in equation (\ref{eq:ez-bound}).  The upper bounds are derived from equation (\ref{eq:benchmark}). Hence,  $(1-R)J(x)>0$ for all $x>0$. 

If $R\in(0,1)$, since $\mathcal{A}(0)$ contains only one element $(0,0)$, then $J(0)=V^{c=0}_0=0$. In addition, $J_x'(0)=+\infty$ can easily be obtained from inequality \eqref{eq:ez-bound}.
\end{proof}

\begin{proposition}\label{pro:ez-concave}  The optimal value function $J$ is  strictly concave on the domain $\mathbb{R}_{++}$.
\end{proposition}
\begin{proof}
Since $g$ is a concave function, then
it is evident that $\lambda \mathcal{C}(x_{1})+(1-\lambda)\mathcal{C}(x_{2})\subset \mathcal{C}(\lambda x_{1}+(1-\lambda) x_{2})$ for $x_{1}, x_{2}>0$ and $\lambda\in(0,1)$.   For any $c^{1}\in \mathcal{C}(x_{1})$ and $c^{2}\in \mathcal{C}(x_{2})$,  we establish that:
\begin{align}  
V_{t}^{c^{\lambda}}\geq \lambda V_{t}^{c^{1}}+(1-\lambda) V_{t}^{c^{2}},~~~~~t\geq0, 
\end{align}where $c^{\lambda} \equiv \lambda c^{1}+(1-\lambda)c^{2}$. 
By the Definition 6.4 and Theorem 6.5 of  \cite{HHJ23b}, it suffices to consider the situations where $\mathbb{E}\left[\int_{0}^{\infty}e^{-\delta s} |f(c^{i}_{s}, V_{s}^{c^{i}})| ds\right]<+\infty$ where $i=1,2,\lambda$, and  $V^{c^{1}}$, $V^{c^{2}}$  and $V^{c^{\lambda}}$  are RCLL.  We analyze the following two cases separably. 

{\em Case 1:  $R>1$ and $S > 1$.}  In this case, $f(c, v)$ is joint concave in $(c,v)$. Then, for any $t\geq0$,  we have: \begin{align*} &V_{t}^{c^{\lambda}}-[\lambda V_{t}^{c^{1}}+(1-\lambda) V_{t}^{c^{2}}]\\=&\mathbb{E}\left[\int_{t}^{\infty}e^{-\delta s}\left( f(c^{\lambda}_{s}, V_{s}^{c^{\lambda}})-\lambda  f(c^{1}_{s}, V_{s}^{c^{1}})- (1-\lambda) f(c^{2}_{s}, V_{s}^{c^{2}})\right)ds~\big|~\mathcal{F}_{t}\right]\\\geq & \mathbb{E}\left[\int_{t}^{\infty}e^{-\delta s}\left( f(c^{\lambda}_{s}, V_{s}^{c^{\lambda}})- f(c^{\lambda}_{s}, \lambda V_{s}^{c^{1}}+(1-\lambda) V_{s}^{c^{2}}) \right)ds~\big|~\mathcal{F}_{t}\right]\\\geq & \mathbb{E}\left[\int_{t}^{\infty} e^{-\delta s} \frac{\partial f}{\partial v}(c^{\lambda}_{s}, V_{s}^{c^{\lambda}}) \left( V_{s}^{c^{\lambda}}-[\lambda V_{s}^{c^{1}}+(1-\lambda) V_{s}^{c^{2}}] \right) ds~\big|~\mathcal{F}_{t}\right].\end{align*}Given that  $\frac{\partial f}{\partial v}(c^{\lambda}_{s}, V_{s}^{c^{\lambda}})\leq 0$, it follows from Lemma C3\footnote{After careful verifications,   Lemma C3 of \cite{SS99} remains valid when $T=+\infty$.}   in \cite{SS99} that  $V_{t}^{c^{\lambda}}-[\lambda V_{t}^{c^{1}}+(1-\lambda) V_{t}^{c^{2}}]\geq 0$ for all $t\geq0$.   Finally, taking the supremum over all feasible consumption strategies, this implies: $J(\lambda x_{1}+ (1-\lambda)x_{2})\geq \lambda J( x_{1})+ (1-\lambda)J(x_{2})$. 

{\em Case 2: $R$, $S\in (0,1)$.} In this situation, $f(c, v)$ is no longer joint concave in $(c,v)$. We approach the proof applying a transformation method as follows.  

Since $\rho<0$ and $R\in(0,1)$, and for $i=1,2,\lambda$,  $\mathbb{E}\left[\int_{0}^{\infty}e^{-\delta s} f(c^{i}_{s}, V_{s}^{c^{i}}) ds\right]<+\infty$, 
 it follows that $V^{c^{i}}_{t}$ is strictly positive for all $t \in (0,+\infty)$ and  $V^{c^{i}}_{\infty}=0$.  By the representation theorem (Chapter 3, Corollary 3, p189) of \cite{P05}, we have: 
$$V_{t}^{c^{i}}=\int_{t}^{+\infty}e^{-\delta s}f(c^{i}_{s}, V_{s}^{c^{i}})ds-\int_{t}^{+\infty}Z_{s}^{i}dB_{s},$$where $\int_{0}^{+\infty}(Z_{s}^{i})^{2}ds<+\infty$ for $i=1,2,\lambda$. 

For $i=1,2,\lambda$,  we define $\tilde{V}_{t}^{c^{i}}=(V_{t}^{c^{i}})^{1-\rho}$ and $\tilde{Z}^{i}_{t}=(1-\rho)(V_{t}^{c^{i}})^{-\rho}Z_{t}^{i}$. Applying  It\^{o}'s formula, we obtain:
$$\tilde{V}_{t}^{c^{i}}=\int_{t}^{+\infty}F(s, c^{i}_{s}, \tilde{V}_{s}^{c^{i}}, \tilde{Z}_{s}^{i})ds-\int_{t}^{+\infty}\tilde{Z}_{s}^{i}dB_{s},$$where
$$F(t,c,\tilde{v}, \tilde{z})=\frac{(1-\rho)(1-R)^{\rho}}{1-S}e^{-\delta t}c^{1-S}+\frac{\rho}{2(1-\rho)}\frac{|\tilde{z}|^{2}}{\tilde{v}}.$$  Since $\rho<0$ and $S, R\in(0,1)$, the function  $F(t,c,\tilde{v}, \tilde{z})$ is jointly concave in $(c,\tilde{v}, \tilde{z})$. 

Let $\Delta \tilde{V} :=\lambda \tilde{V}^{c^{1}}+(1-\lambda) \tilde{V}^{c^{2}}$ and $\Delta \tilde{Z} :=\lambda \tilde{Z}^{1}+(1-\lambda) \tilde{Z}^{2}$. Then, we have: 
$$\Delta \tilde{V}_{t}=\int_{t}^{+\infty} (F(s,c_{s}^{\lambda}, \Delta \tilde{V}_{s}, \Delta \tilde{Z}_{s})+A_{s})ds-\int_{t}^{+\infty}\Delta \tilde{Z}_{s}dB_{s},$$where $A_{s}=\lambda F(s, c^{1}_{s}, \tilde{V}_{s}^{c^{1}}, \tilde{Z}_{s}^{1})+(1-\lambda)F(s, c^{2}_{s}, \tilde{V}_{s}^{c^{2}}, \tilde{Z}_{s}^{2})-F(s,c_{s}^{\lambda}, \Delta \tilde{V}_{s}, \Delta \tilde{Z}_{s})\leq0$.  Moreover, applying It\^o's formula gives: 
$$(\Delta \tilde{V}_{t})^{\frac{1}{1-\rho}}=\mathbb{E}\left[\int_{t}^{+\infty}e^{-\delta s}f(c_{s}^{\lambda}, (\Delta \tilde{V}_{s})^{\frac{1}{1-\rho}})+\frac{(\Delta \tilde{V}_{s})^{\frac{\rho}{1-\rho}}}{1-\rho}A_{s}ds~\big|~\mathcal{F}_{t}\right],$$
Therefore, it follows that 
\begin{align*}
V_{t}^{c^{\lambda}}-(\Delta \tilde{V}_{t})^{\frac{1}{1-\rho}}\geq &\mathbb{E}\left[\int_{t}^{+\infty}e^{-\delta s}\left( f(c^{\lambda}_{s}, V_{s}^{c^{\lambda}})-f(c_{s}^{\lambda}, (\Delta \tilde{V}_{s})^{\frac{1}{1-\rho}})\right)ds~\big|~\mathcal{F}_{t}\right]\\
\geq &\mathbb{E}\left[\int_{t}^{+\infty}e^{-\delta s} \frac{\partial f}{\partial v}(c_{s}^{\lambda}, (\Delta \tilde{V}_{s})^{\frac{1}{1-\rho}}) \left(V_{s}^{c^{\lambda}}-(\Delta \tilde{V}_{s})^{\frac{1}{1-\rho}}\right)ds~\big|~\mathcal{F}_{t}\right],
\end{align*}
where the second inequality follows from the fact that $f(c,v)$ is convex in $v$.  Moreover, since $\frac{\partial f}{\partial v}\leq0$, it follows from Lemma C3 in \cite{SS99} and $\rho<0$ again that:
$$V_{t}^{c^{\lambda}}\geq(\Delta \tilde{V}_{t})^{\frac{1}{1-\rho}}\geq \lambda (\tilde{V}_{t}^{c^{1}})^{\frac{1}{1-\rho}}+(1-\lambda) (\tilde{V}_{t}^{c^{2}})^{\frac{1}{1-\rho}}=\lambda V_{t}^{c^{1}}+(1-\lambda) V_{t}^{c^{2}},$$
where the second inequality follows from the inequality $(\lambda x+(1-\lambda) y)^{p}\geq \lambda x^{p}+(1-\lambda) y^{p}$ for any $p,\lambda\in(0,1)$ and $x,y\geq0$.

Thus, we conclude that the optimal value function $J$ is concave.  Finally, if the function $J$ is not strictly concave in certain periods, then it must be constant, which contradicts \eqref{eq:ez-bound} in Proposition \ref{pro:basic pro}.
\end{proof}

\begin{proposition}
    \label{pro:ez-continuous}
    For $0<R<1$, the optimal value function is uniformly continuous. For $R>1$, the optimal value function is uniformly continuous on $[a, +\infty)$ for any positive number $a$.
\end{proposition}
\begin{proof}
The uniformly continuous property on $\mathbb{R}_+$ can be derived similar to Proposition 2.2 in \cite{Z94}. For $0 < R < 1$, from \eqref{eq:ez-bound}, together with  $J(0)=0$ and the concavity of $J$, it follows that $J$ is uniformly continuous on bounded interval $[0,M]$ with  $M>0$. For the behavior on infinite intervals, uniform continuity follows from the vanishing difference quotient implied by the sublinear growth property. The case $R>1$ can be obtained in a similar way.
\end{proof}

\subsection{Dynamic programming principle }
In this subsection, we develop the dynamic programming principle (DPP), Proposition \ref{pro: dpp}, within our framework. From this subsection onward, we will focus on analyzing the case where $0<R<1$; for the case where 
$R>1$, the detailed discussion can be found in Remark \ref{rem:r>1}, Remark \ref{rem:r>1 uni}, and the Appendix D. 

The approach is as follows. First, we establish a relationship between Epstein–Zin utility and an infinite-horizon BSDE. Next, we use the BSDE to define a family of backward semigroups. Finally, we formulate the DPP in the case where the consumption process is truncated. This truncation is an essential step in handling the non-Lipschitz nature of the aggregator.

We now establish the relationship between utility function $V$ and the following infinite horizon BSDE:
\begin{eqnarray}\label{bsde1}
Y_t-Y_T=\int^{T}_{t}[-\delta \nu Y_s+f(c_{s},Y_s)]ds
 -\int^{T}_{t}Z_sdB_s, \quad 0\leq t< T<+\infty.
\end{eqnarray}
\begin{lemma}\label{lem:bsde1}
For any $c\in \mathcal{C}(x)$, BSDE (\ref{bsde1}) admits a unique solution $(Y_t^c,Z_t^c)_{t\geq0}$ with $Y^c_t=e^{\delta\nu t}V_t^c$ and $\int^\infty_0e^{-2\delta\nu t}|Z_t^c|^2dt< +\infty$.
\end{lemma}

\begin{proof}
By Proposition \ref{pro:basic pro},
for any $c\in\mathcal{C}(x)$, $$
\mathbb{E}[V^c_0(x)]\leq J(x)<+\infty.
$$
Then, by the martingale representation theorem, there exists a unique $(\mathcal{F}_t)$-adapted process  $(Z^1_t)_{t\geq0}$ such that $\int^{\infty}_{0}|Z^1_s|^2ds<+\infty$
and
$$
\mathbb{E}\left[\int_{0}^{\infty}e^{-\delta s} f(c_{s},V_{s}^{c})ds~\big|~\mathcal{F}_{t}\right]
=\mathbb{E}\left[\int_{0}^{\infty}e^{-\delta s} f(c_{s},V_{s}^{c})ds\right]+\int^t_0Z^1_sdB_s,
$$
i.e.,
$$
\mathbb{E}\left[\int_{t}^{\infty}e^{-\delta s} f(c_{s},V_{s}^{c})ds~\big|~\mathcal{F}_{t}\right]
=\mathbb{E}\left[\int_{0}^{\infty}e^{-\delta s} f(c_{s},V_{s}^{c})ds\right]+\int^t_0Z^1_sdB_s
-\int_{0}^{t}e^{-\delta s} f(c_{s},V_{s}^{c})ds.
$$
Then
$$
 V_t^c-V_T^c=\int_{t}^{T}e^{-\delta s} f(c_{s},V_{s}^{c})ds
 -\int^T_tZ^1_sdB_s,\ \ 0\leq t< T<+\infty.
$$
Applying It\^o's formula to $e^{\delta \nu t}V_t^c$,
$$
 de^{\delta \nu t}V_t^c=[\delta \nu e^{\delta \nu t}V_t^c-e^{-\delta t}e^{\delta \nu t}f(c_{t},V_{t}^{c})]dt
 +e^{\delta \nu t}Z^1_tdB_t.
$$
Then
$(e^{\delta\nu  t}V_t^c, Z^c_t)_{t\geq0}:=(e^{\delta\nu  t}V_t^c, e^{\delta \nu t}Z^1_t)_{t\geq0}$
 is the  solution of infinite horizon BSDE (\ref{bsde1}).
 
 The uniqueness of $(Y^c, Z^c)$ follows from the uniqueness of $V^c$. In fact, assume that $(Y^{1,c}, Z^{1,c})$ is another solution of BSDE (\ref{bsde1}). Applying It\^o's formula to $e^{-\delta \nu t}Y_t^{1,c}$ and then taking conditional expectation with respect to $\mathcal{F}_t$,
 $$
  e^{-\delta \nu t}Y^{1,c}_t-\mathbb{E}\left[e^{-\delta \nu T}Y^{1,c}_T|\mathcal{F}_t\right]=\mathbb{E}\left[\int_{t}^{T}e^{-\delta s} f(c_{s},e^{-\delta \nu s }Y^{1,c}_s)ds|\mathcal{F}_t\right],\ \ 0\leq t< T<+\infty.
 $$
 Letting $T\rightarrow +\infty$, we have 
  $$
  e^{-\delta \nu t}Y^{1,c}_t=\mathbb{E}\left[\int_{t}^{\infty}e^{-\delta s} f(c_{s},e^{-\delta \nu s }Y^{1,c}_s)ds|\mathcal{F}_t\right],\ \ 0\leq t<+\infty,
 $$
 and $Y^{1,c}_t=e^{\delta\nu  t}V_t^c=e^{\delta\nu  t}V_t^c$ and  $Z^{1,c}_t=Z^{c}_t$.
 \end{proof}
 

Let us study the well-posedness of the finite horizon BSDE for some fixed $T$:
\begin{align}\label{eq:bsde-fix-terminal}
 V_t=\zeta+\int_{t}^{T}e^{-\delta s} f(c_{s},V_{s})ds
 -\int^T_tZ_sdW_s,\ \ 0\leq t\leq T.
\end{align}
\begin{lemma}\label{lem:bsde-rs<1}
 For any fixed finite horizon $T>0$, $0<\zeta\in \mathcal{F}_{T}$ and $c \in {\cal L}_{+}$  such that $\mathbb{E}[\zeta^{\frac{1}{\nu}}+\int_{0}^{T}e^{-\delta s}c_{s}^{1-S}ds]<+\infty$, then the finite horizon BSDE
  (\ref{eq:bsde-fix-terminal}) 
admits a unique solution $(V, Z)$ such that $V$ is strictly positive, and of {class (D)}, with $\int_{0}^{T}Z_{s}^{2}ds<+\infty$. Moreover,  for any  $\zeta_{i}$ and $c^{i}$, $i=1,2$, such that $\zeta_{1}\geq \zeta_{2}$ and $c^{1}\geq c^{2}$, then we have that $V_{t}^{{1}}\geq V_{t}^{{2}}$. In particular, $V_{t}^{{1}}- V_{t}^{{2}}\leq \mathbb{E}[\zeta_{1}-\zeta_{2}|\mathcal{F}_t]$ when $c^{1}= c^{2}$.
\end{lemma}
\begin{proof}
Similar to Lemma 8.1 in \cite{hu24}, one can construct the generator  as follows:  
 $$f_{m}(c,v)=\frac{(c\wedge m)^{1-S}}{1-S}\left((1-R)v\vee \frac{1}{m}\right)^{\rho},$$
 where $m > 0$. Under this construction,  $f_{m}$ is globally Lipschitz continuous in $v$. Moreover,
 $0\leq f_{m}\leq \frac{m^{1-S-\rho}}{1-S}$.
 Then the truncated BSDE:
 \begin{align}\label{eq:bsde-fix-terminal-m}
 V_t^{c,m}=\zeta\wedge m+\int_{t}^{T}e^{-\delta s} f_{m}(c_{s},V_{s}^{c,m})ds
 -\int^T_tZ^{c,m}_sdW_s,\ \ 0\leq t\leq T,
\end{align}admits a unique solution $(V^{c,m}, Z^{c,m})$. Besides, $f_{m}$ and $\zeta\wedge m$ are both increasing with respect to $m$.  By the comparison theorem,  it follows that $V^{c,m}$ is increasing with $m$, and nonnegative (taking generator and terminal random value as $0$).

To finish the proof, we now establish an upper bound for $V^{c,m}$ that is independent of $m$.
Applying It\^o's formula to $(V^{c,m})^{\frac{1}{\nu}}$, taking the conditional expectation, one has that
\begin{align*}
(V_{t}^{c,m})^{\frac{1}{\nu}}&\leq\mathbb{E}\left[ (\zeta\wedge m)^{\frac{1}{\nu}}+\int_{t}^{T}\frac{1}{\nu} (V_{s}^{c,m})^{-\rho} e^{-\delta s}f_{m}(c_{s}, V_{s}^{c,m})~\big|~\mathcal{F}_{t}\right]\\
&\leq \mathbb{E}\left[ \zeta^{\frac{1}{\nu}}+(1-R)^{\rho-1}\int_{t}^{T}e^{-\delta s}c_{s}^{1-S}ds ~\big|~\mathcal{F}_{t}\right].
\end{align*}
Hence,
$$V_{t}^{c,m}\leq  \left(\mathbb{E}\left[ \zeta^{\frac{1}{\nu}}+(1-R)^{\rho-1}\int_{t}^{T}e^{-\delta s}c_{s}^{1-S}ds ~\big|~\mathcal{F}_{t}\right]\right)^{\nu}=:\bar{V}^{c}.$$
Then, by the localization argument in \cite{bh06}, one obtains the existence of a solution to the BSDE \eqref{eq:bsde-fix-terminal}.
  Since $f(c,v)$ is decreasing in $v$, uniqueness can similarly be obtained by following the procedure of Proposition 2.2 in \cite{X17}. 
\end{proof}


For any $N>0$, define
 \begin{equation}\label{valueN}
J^N(x)=\sup _{(\pi, c) \in \mathcal{A}^N(x)} V_{0}^{c}=\sup _{c \in \mathcal{C}^N(x)} V_{0}^{c},
~~~~x>0,
\end{equation}
where$${\mathcal{A}}^N(x):=\{(\pi,c)\in \mathcal{A}(x)~|~
 c_t\leq N, ~\forall t \ge 0\},\ \ N>0,$$
and the consumption stream $c\in \mathcal{C}^N(x)$ if there is an investment process $\pi$ such that $(\pi,c)\in \mathcal{A}^N(x)$.

Let $X^{x,\pi,c}$ denote the solution of (\ref{eq:wealth}). For any finite horizon $T>0$,  notice that $\mathbb{E}[(J^N(X^{x,\pi,c}_T))^{\frac{1}{\nu}}]<+\infty$ for any $(\pi,c)\in\mathcal{A}^N(x)$. By Lemma \ref{lem:bsde-rs<1}, the finite-horizon BSDE
\begin{eqnarray}\label{BSDE-J_N}
\begin{cases}
dY^{\pi,c,J^N}_s=[\delta \nu Y^{\pi,c,J^N}_s-f(c_{s},Y^{\pi,c,J^N}_s)]ds+Z^{\pi,c,J^N}_sdB_s, \\
Y^{\pi,c,J^N}_{ T}=J^N(X^{x,\pi,c}_{T}),
\end{cases}
\end{eqnarray}
admits a unique solution $(Y^{\pi,c,J^N},Z^{\pi,c,J^N})$. 
  Then, we can define the family of backward semigroups, motivating by \cite{peng}:
  \begin{eqnarray}\label{gdpp}
G^{c}_{s,T}[J^N(X^{x,\pi,c}_{T})]:=Y^{\pi,c,J^N}_s,\ \
\ \ \ \ s\in[0,T].
\end{eqnarray}

We are now ready to state and prove the DPP.
\begin{proposition}\label{pro: dpp}
For any $N>0$, finite time $T>0$ and $x\geq 0$,  and stopping times  $0< \tau\leq T$, then the following DPP hold:
\begin{align}\label{dppzhou}
J^N(x)=\sup_{{(\pi,c)\in\mathcal{A}}^N(x)}G^{c}_{0, \tau}[J^N(X^{x,\pi,c}_{\tau})].
\end{align}
\end{proposition}
\begin{proof}
 Let $\bar{J}^N(x):=\sup_{(\pi,c)\in{\mathcal{A}^N}(x)}G^{c}_{0, \tau}[J^N(X^{x,\pi,c}_{\tau})]$. We need to prove that $$J^N(x)=\bar{J}^N(x).$$
 
 Without loss of generality, we can suppose that $(\Omega,(\mathcal{F}_t)_{t\geq0},\mathcal{F}, \mathbb{P})$ is the standard Wiener space, $\Omega=\{\omega \in C(\mathbb{R}_+;\mathbb{R}): \omega_0=0\}$, $B$ the canonical process, $\mathbb{P}$ the Wiener measure, $\mathcal{F}$ the Borel $\sigma$-field over $\Omega$, completed  with respect to the Wiener measure $\mathbb{P}$ on this space. For $t>0$, let $\theta_{t}=\theta_{t}(\omega)$ be the translation operator on $\Omega$:
$$
                       \theta_{t}(\omega)_s=\omega(s+t)-\omega(t), \ \ \omega\in \Omega,\ s\geq0.
$$
Given $(\pi,{c})\in {\mathcal{M}}\times\mathcal{L}_+ $ we can identify $(\pi,c)$ with a  measurable mapping applying to $B$. Thus, for any $(\pi^0,c^0)\in \mathcal{M}\times\mathcal{L}_{+}$, we can define
$$\bar{\pi}_s:=\begin{cases} \pi^0, \ \ s\in [0,t),\\
\pi_{s-t}(\theta_{t}), \ \ s\geq t,
\end{cases} \ \ \
\bar{c}_s:=\begin{cases} c^0, \ \ s\in [0,t),\\
c_{s-t}(\theta_{t}), \ \ s\geq t.
\end{cases}
$$
Then, $(\bar{\pi},\bar{c})\in \mathcal{M}\times\mathcal{L}_+$.
Let $t>0,x>0,$ we define $\mathcal{A}(t,x)$  in the spirit of $\mathcal{A}(x)$ and $\mathcal{C}^N(t,x)$ in the spirit of $\mathcal{C}^N(x)$. Let  $(\pi,c)\in \mathcal{A}(t,x)$ and $X^{t,x,\pi,c}$ be the unique solution of SDE:
\begin{eqnarray}\label{wealthtau}
         X_s=x+\int^{ s}_{t}rX_l+(\mu-r)\pi_l-c_ldl+\int^{s}_{t}\pi_l\sigma dB_l, \ \ s\geq t, x\geq0,
\end{eqnarray}
and $(Y^{t,x,c},Z^{t,x,c})$ be the unique solution of BSDE, for $t\leq s\leq T<+\infty$, 
\begin{eqnarray}\label{bsdetau}
Y^{t,c}_s=Y^{t,c}_T+\int^{T}_{s}f(c_l,Y^{t,c}_l)-\delta\nu Y^{t,c}_ldl-\int^{T}_{s}Z^{t,c}_ldB_l.
\end{eqnarray}
Applying the transformation $\theta_{t}$ to (\ref{eq:wealth}) and (\ref{bsde1}), we have that $(X^{x,\pi,c}_{s-t}(\theta_{t}),Y^{c}_{s-t}(\theta_{t}),Z^{c}_{s-t}(\theta_{t}))_{s\geq t}$ are the solution of (\ref{wealthtau}) and (\ref{bsdetau}) with control process $(\bar{\pi},\bar{c})$ instead of $(\pi,c)$ (See Lemma A.1 in \cite{LZ}). It is clear that $(\pi_{\cdot-t}(\theta_t),c_{\cdot-t}(\theta_{t}))\in \mathcal{A}(t,x)$ if $(\pi,c)\in\mathcal{A}(x)$, then by the uniqueness of the solutions of (\ref{wealthtau}) and (\ref{bsdetau}), we have
$$X^{x,\pi,c}_s(\theta_{t})=X^{t,x,\bar{\pi},\bar{c}}_{s+t},\ \ Y^{c}_s(\theta_{t})=Y^{t,\bar{c}}_{s+t},\ \mathbb{P}\mbox{-a.s.}, \ \  Z^{c}_s(\theta_{t})=Z^{t,\bar{c}}_{s+t}, \ \ ds\times d\mathbb{P}\mbox{-a.e.}, \ s\geq 0.
$$
 Since $J^N$ is deterministic function, we have, for every $t\in [0,T]$, 
$$J^N(y)=\mathop{\esssup}\limits_{v\in{\mathcal{C}^N(y)}}Y_0^{v}(\theta_{ t})
=\mathop{\esssup}\limits_{v\in{\mathcal{C}^N}(y)}Y^{ t,\bar{v}}_{t},\ \  \mathbb{P}\mbox{-a.s.},
 $$
 with
 $$\bar{v}_s:=\begin{cases} v^0,\ \ \ s\in [0, t),\\
 v_{s- t}(\theta_{t}),\ \ s\geq t.
\end{cases}
 $$
 Then, by the standard argument (see, e.g., \cite{peng}),
\begin{align*}
J^N(X^{c,\pi,x}_{t})=\mathop{\esssup}\limits_{v\in{\mathcal{C}^N}(x)}Y_{t}^{t,\bar{v}}
=\mathop{\esssup}\limits_{v\in{\mathcal{C}}^N(x)}Y_{t}^{c\oplus\bar{v}},
\end{align*}
 where
 $$
 (c\oplus\bar{v})_s=\begin{cases} c_s \ \ \ s\in [0,t)\\
                                \bar{v}_s,\ \ s\geq t
 \end{cases}\in \mathcal{C}^N(x).
 $$
 Again by the standard argument (see, e.g., \cite{peng}), for every $\varepsilon>0$, there exists $v\in \mathcal{C}^N(x)$ such that $c=v$, $ds\times d\mathbb{P}$-a.e. on $[0,t]\times\Omega$ and
 $$
              J^N(X^{x,\pi,c}_{t})\leq Y_{t}^{v}+\varepsilon,\ \ \mathbb{P}\mbox{-a.s.}
 $$
 From Lemma \ref{lem:bsde-rs<1},
 \begin{align*}
  G^{c}_{0,t}[J^N(X^{x,\pi,c}_{t})]\leq G^{c}_{0,t}[Y_{t}^{v}+\varepsilon]
  \leq G^{c}_{0,t}[Y_{t}^{v}]+\varepsilon&=Y_0^{v}+\varepsilon\\
  &\leq \mathop{\esssup}\limits_{c\in {\mathcal{C}}^N(x)} Y_0^{c}+\varepsilon=J^N(x)+\varepsilon,\ \ \mathbb{P}\mbox{-a.s.}
 \end{align*}
 By the arbitrariness of $\varepsilon$, we see that
 $$
 G^{c}_{0,t}[J^N(X^{x,\pi,c}_{t})]\leq J^N(x).
 $$
 Let stopping times sequence $\{\tau_n\}_{n\geq1}$ be the simple approximation of stopping times $\tau$, then we have
  $$
 G^{c}_{0,\tau_n}[J^N(X^{x,\pi,c}_{\tau_n})]\leq J^N(x),\ \ n\geq1.
 $$
 Letting $n\rightarrow\infty$, it follows that
 \begin{align*}
 G^{c}_{0,\tau}[J^N(X^{x,\pi,c}_{\tau})]\leq J^N(x) \mbox{~~and~~~} \bar{J}^N(x)=\sup_{(\pi,c)\in{\mathcal{A}}^N(x)}G^{c}_{0,\tau}[J^N(X^{x,\pi,c}_{\tau})]\leq J^N(x).
 \end{align*}
 To prove $J^N(x)\leq \bar{J}^N(x)$, let, for any given $\varepsilon>0$, $(\pi,c)\in \mathcal{A}^N(x)$ be such that $J^N(x)\leq Y^{c}_0+\varepsilon$. Then, for every $t\in [0,T]$,
 \begin{align*}
 J^N(x)\leq  Y^{c}_0+\varepsilon=G^{c}_{0,t}[Y^{c}_{t}]+\varepsilon
   &\leq G^{c}_{0,t}[\mathop{\esssup}\limits_{\bar{v}\in{\mathcal{C}}^N(x)}Y^{c\oplus\bar{v}}_{t}]+\varepsilon\\
   &= G^{c}_{0,t}[\mathop{\esssup}\limits_{\bar{v}\in{\mathcal{C}}^N(t,X^{x,\pi,c}_{t})}Y^{t,{\bar{v}}}_{t}]
       +\varepsilon, \ \ \mathbb{P}\mbox{-a.s.}
 \end{align*}
 Since $Y^{t,\bar{v}}_{t}=(Y^{{v}}_0)(\theta_{t})=Y^{{v}}_0
 $,
 $$
      \mathop{\esssup}\limits_{\bar{v}\in{\mathcal{C}}^N(t,X^{x,\pi,c}_{t})}Y^{t,
      \bar{v}}_{t}
      =
\mathop{\esssup}\limits_{\bar{v}\in{\mathcal{C}}^N(t,y)}Y^{t,
      \bar{v}}_{t}|_{y=X^{x,\pi,c}_{t}}
      =\mathop{\esssup}\limits_{{v}\in{\mathcal{C}}^N(y)}Y^{{v}}_0
     |_{y=X^{x,\pi,c}_{t}}
      =J^N(X^{x,\pi,c}_{t}).
 $$
 Consequently,
 $$
 J^N(x)\leq  G^{c}_{0,t}[J_N(X^{x,\pi,c}_{t})] +\varepsilon
 \leq \sup_{c\in{\mathcal{A}}^N(x)}G^{c}_{0,t}[J^N(X^{x,\pi,c}_{t})] +\varepsilon,\ \ \mbox{and}
 $$
 $$
 J^N(x)
 \leq \sup_{c\in{\mathcal{C}}^N(x)}G^{c}_{0,\tau_n}[J^N(X^{x,\pi,c}_{\tau_n})] +\varepsilon.
 $$
 Letting $n\rightarrow\infty$ and  $\varepsilon\rightarrow0^+$, we obtain that $J^N(x)\leq \bar{J}^N(x)$.
\end{proof}
\begin{rem}
(i) Because of the exponential decay function $e^{-\delta s}$, $V^c$ defined in (\ref{eq:ez-utility}) is time inhomogeneous. Therefore, it is difficult to directly establish the DPP for the value function $J$.

(ii) The literature (e.g., \cite{LZ}) establishes the DPP for the value function defined by the solution of a time-homogeneous infinite horizon BSDE. To obtain the DPP for $J$, we need to find a suitable time-homogeneous infinite horizon BSDE (\ref{bsde1}) and establish a one-to-one correspondence between its solution $Y^c$ and $V^c$ defined in (\ref{eq:ez-utility}).

(iii)
       Unlike standard literature (e.g., \cite{EPQ97}), we can not prove the existence and uniqueness of (\ref{BSDE-J_N}) with terminal condition $Y_t^{c,\pi,J}=J(X^{x,\pi,c}_t)$,  since $f$ does not satisfy Lipschitz  condition for some $(\pi,c)\in \mathcal{A}(x)$. Therefore, we cannot directly provide the DPP of value function $J$. However, the DPP for the approximation  value function $J^N$ is sufficient to ensure the existence of viscosity solutions to the HJB equation (\ref{eq:HJB}), which we address in Subsection \ref{subsec:exist}.
\end{rem}

\section{The Proof of Theorem \ref{th-unique}}
\label{sec:viscosity}

The objective of this section is to prove Theorem \ref{th-unique}. We first show that the optimal value function is indeed a viscosity solution to the HJB equation (\ref{eq:HJB}), as stated in Theorem \ref{existence}. This proof relies heavily on the dynamic programming principle. In the second part of the section, we establish that the viscosity solution to the HJB equation (\ref{eq:HJB}) in Theorem \ref{th-unique} is indeed unique.
\subsection{Existence of viscosity solutions}\label{subsec:exist}
In this subsection,  we consider the existence of viscosity solution to the  HJB equation (\ref{eq:HJB}).

For any $x\in \mathbb{R}_+$ and $w \in C(\mathbb{R}_+)$, define 
\begin{eqnarray*}
  \mathcal{A}^+(x,w):=\bigg{\{}&&\varphi\in C^{2}(\mathbb{R}_{+}):  0={w}(x)-\varphi(x)=\sup_{y\in \mathbb{R}_+}
                         ({w}(y)- \varphi(y
                         )),\ \mbox{and} \\
                       &&  \ \varphi\in C^{2}_b([a,\infty)) \ \mbox{for all} \ a>0 \bigg{\}},
\end{eqnarray*}

and
\begin{eqnarray*}
 \mathcal{A}^-(x,w):=\bigg{\{}&&\varphi\in C^{2}(\mathbb{R}_+):  0={w}(x)+\varphi(x)=\inf_{y\in \mathbb{R}_+}
                         ({w}(y)+\varphi(
                         y)),\ \mbox{and} \\
                       && \ \varphi\in C^{2}_b([a,\infty)) \ \mbox{for all} \ a>0  \bigg{\}}.
\end{eqnarray*}

\begin{definition}\label{defviscosity} \ \
 A function $w\in C(\mathbb{R}_+)$ is called a
                             viscosity subsolution (resp.,  supersolution)
                             to  equation (\ref{eq:HJB}) if whenever  $\varphi\in {\cal{A}}^+(x,w)$ (resp.,  $\varphi\in {\cal{A}}^-(x,w)$)  with $x\in \mathbb{R}_{++}$,  we have
\begin{eqnarray*}
                          -\delta\nu w(x)
                           +{\mathbf{H}}(x,w(x),\varphi_{x}(x),\varphi_{xx}(x))\geq0,
\end{eqnarray*}
\begin{eqnarray*}
                          (\mbox{resp.},\ -\delta\nu w(x)
                           +{\mathbf{H}}(x,w(x),-\varphi_{x}(x),-\varphi_{xx}(x))
                          \leq0).
\end{eqnarray*}
                                $w\in C(\mathbb{R}_+)$ is said to be a
                             viscosity solution to equation (\ref{eq:HJB}) if it is
                             both a viscosity subsolution and a viscosity
                             supersolution.
\end{definition}

There are three steps in proving the existence theorem. First,  we  show that $J^N(x)$ is a viscosity solution to the following approximation  HJB equation:
 \begin{align}
\label{eq:HJBN}
-\delta\nu J^N(x)+\mathbf{H}_N(x,J^N(x),J^N_x(x),J^N_{xx}(x))=0,  \ x\in \mathbb{R}_+,
\end{align}
where
\begin{eqnarray*}
\mathbf{H}_N(x,k,p,q)&=&\sup_{{|\pi|\leq g(x)}}[\pi(\mu-r)p+\frac{1}{2}\sigma^2\pi^2q]+\sup_{0\leq c\leq N}[f(c,k)-cp]+rxp,\\
&&\ \ \ \ \ \ \ \ \ \  \ \ \ \ \ \ \ \ \  \ \ \ \  \ \ \  \ \ \ \ \ \ \ \ \ \ \ \ \ (x,k,p,q)\in
\mathbb{R}_+\times \mathbb{R}_+\times \mathbb{R}\times \mathbb{R}.
\end{eqnarray*}
\begin{theorem}\label{existenceN}
 The value function $J^N(x)$ defined by (\ref{valueN}) is a viscosity solution to equation  (\ref{eq:HJBN}).
\end{theorem}
\begin{proof}
First, let  $\varphi\in {\cal{A}}^+(\hat{x},J_N)$ with $\hat{x}\in \mathbb{R}_{++}$, and define
$$\tau^{\pi,c}=\inf\left\{t~|~X^{\hat{x},\pi,c}_t<\frac{\hat{x}}{2}\right\},\ \ (\pi,c)\in \mathcal{A}^N(\hat{x}).$$
                   Then by the DPP ( Proposition \ref{pro: dpp}), we obtain the following result: for any $h>0$,
 \begin{eqnarray}\label{4.91}
 \begin{aligned}
                            0&=J^N(\hat{x})-\varphi (\hat{x})\\
                           &=\sup_{(\pi,c)\in {\mathcal{A}}^N(\hat{x})} \mathbb{E}\left[J^N(X^{\hat{x},\pi,c}_{\tau^c\wedge h})
                           +\int_0^{\tau^c\wedge h}[f(c_s,Y^{c,J^N}_s)-\delta \nu Y^{c,J^N}_s]ds\right]
                           -\varphi (\hat{x}).
                           \end{aligned}
\end{eqnarray}
Then, for any $\varepsilon>0$ and $h>0$,  we can  find a control  $(\pi^{\varepsilon},c^{\varepsilon})\equiv (\pi^{\varepsilon,h},c^{\varepsilon,h})\in {\mathcal{A}}^N(\hat{x})$ such
   that the following result holds:
\begin{eqnarray}\label{4.10}
\begin{aligned}
    -{\varepsilon}h
    &\leq\mathbb{E}\left[J^N(X^{\hat{x},\pi^{\varepsilon},c^{\varepsilon}}_{\tau^{c^\varepsilon}\wedge h})
                           +\int_0^{\tau^{c^\varepsilon}\wedge h}[f(c^{\varepsilon}_s,Y^{c^{\varepsilon},J^N}_s)-\delta \nu Y^{c^{\varepsilon},J^N}_s]ds\right]
                           -\varphi (\hat{x})\\
                           &\leq\mathbb{E}\left[\varphi(X^{\hat{x},\pi^{\varepsilon},c^{\varepsilon}}_{\tau^{c^\varepsilon}\wedge h})
                           +\int_0^{\tau^{c^\varepsilon}\wedge h}[f(c^{\varepsilon}_s,Y^{c^{\varepsilon},J^N}_s)-\delta \nu Y^{c^{\varepsilon},J^N}_s]ds\right]
                           -\varphi (\hat{x}).
                           \end{aligned}
\end{eqnarray}
We note that $(Y^{c^{\varepsilon},J^N},Z^{c^{\varepsilon},J^N})$ is the solution of BSDE:
$$
dY^{c^{\varepsilon},J^N}_s=[\delta \nu Y^{c^{\varepsilon},J^N}_s-f(c^{\varepsilon}_{s},Y^{c^{\varepsilon},J^N}_s)]ds
 +Z^{c^{\varepsilon},J^N}_sdB_s, \ \ Y^{c^{\varepsilon},J^N}_{\tau^{c^\varepsilon}\wedge h}=J^N(X^{\hat{x},\pi^\varepsilon,c^{\varepsilon}}_{\tau^{c^\varepsilon}\wedge h}).
$$
 Since $(\pi^{\varepsilon},c^{\varepsilon})\in \mathcal{A}^N(\hat{x})$, for every $p\geq 1$, there exists a constant $C_p>0$, which is  independent of $\varepsilon$, such that
\begin{eqnarray}\label{conti,X}
           \mathbb{E}\left[\sup_{s\in [0,h]}|X^{\hat{x},\pi^{\varepsilon},c^{\varepsilon}}_{s}-\hat{x}|^p\right]\leq C_ph^{\frac{p}{2}}.
\end{eqnarray}
Then, by Propositions \ref{pro:basic pro} - \ref{pro:ez-continuous},
\begin{eqnarray*}
\mathbb{E}|Y^{c^{\varepsilon},J^N}_{\tau^{c^\varepsilon}\wedge h}-J^N(\hat{x})|=\mathbb{E}|J^N(X^{\hat{x},\pi^{\varepsilon},c^{\varepsilon}}_{\tau^{c^\varepsilon}\wedge h})-J^N(\hat{x})|\rightarrow0\ \mbox{as}\ h\rightarrow0^+.
\end{eqnarray*}
Therefore, for $s\leq h$, as $h\rightarrow0^+$, 
\begin{eqnarray}\label{conti,J_N}
\mathbb{E}|Y^{c^{\varepsilon},J^N}_{\tau^{c^\varepsilon}\wedge s}-J^N(\hat{x})|\leq
\mathbb{E}|Y^{c^{\varepsilon},J^N}_{\tau^{c^\varepsilon}\wedge s}-Y^{c^{\varepsilon},J^N}_{\tau^{c^\varepsilon}\wedge h}|
+\mathbb{E}|Y^{c^{\varepsilon},J^N}_{\tau^{c^\varepsilon}\wedge h}-J^N(\hat{x})|\rightarrow0.
\end{eqnarray}
Moreover, from (\ref{conti,X}),
\begin{eqnarray*}
    C_1h^{\frac{1}{2}}\geq \mathbb{E}\left[\sup_{s\in [0,\tau^{c^{\varepsilon}}\wedge h]}|X^{\hat{x},\pi^{\varepsilon},c^{\varepsilon}}_{s}-\hat{x}|\right]
    \geq \frac{\hat{x}}{2}\mathbb{E}[\mathbf{1}_{\{\tau^{c^{\varepsilon}}\leq h\}}]= \frac{\hat{x}}{2}[1-\mathbb{E}[\mathbf{1}_{\{\hat{\tau}^{c^{\varepsilon}}> h\}}]].
\end{eqnarray*}
Thus,
\begin{eqnarray}\label{tau>0}
\frac{\mathbb{E}[\hat{\tau}^{c^{\varepsilon}}\wedge h]}{h}\geq \mathbb{E}[\mathbf{1}_{\{\hat{\tau}^{c^{\varepsilon}}> h\}}]\geq 1-\frac{2C_1h^{\frac{1}{2}}}{\hat{x}}.
\end{eqnarray}
     Applying  It\^{o} formula  to
                         ${\varphi}(X^{\hat{x},\pi^{\varepsilon},{c}^{\varepsilon}}_s)$,   we get that
\begin{eqnarray}\label{bsde4.21}
 \begin{aligned}                            \varphi(X^{\hat{x},\pi^{\varepsilon},c^{\varepsilon}}_{\tau^{c^\varepsilon}\wedge h})
                             =& \varphi(\hat{x})+\int^{\tau^{c^\varepsilon}\wedge h}_{0} ({\cal{L}}\varphi)(X^{\hat{x},\pi^{\varepsilon},{c}^{\varepsilon}}_s,
                                     {\pi}^{\varepsilon}_s, c^{\varepsilon}_s)ds\\
                                     &
                             +\int^{\tau^{c^\varepsilon}\wedge h}_{0}[
                             \pi_s^{\varepsilon}\sigma\partial_x{\varphi}(X^{\hat{x},\pi^{\varepsilon},{c}^{\varepsilon}}_s)]dB_s,
\end{aligned}
\end{eqnarray}
where
\begin{eqnarray*}
                       ({\cal{L}}{\varphi})(x,\pi,c)
                       =
                                         \langle{\varphi}_{x}(x),rx+(\mu-r)\pi-c\rangle+\frac{1}{2}\varphi_{xx}(x)\pi^2\sigma^2, \
                                         ( x,\pi,c)\in \mathbb{R}_+\times \mathbb{R}\times \mathbb{R}.
\end{eqnarray*}
Put (\ref{bsde4.21}) into (\ref{4.10}) and take expectation, by (\ref{tau>0}) we obtain 
\begin{eqnarray*}
                         -{\varepsilon}
                           &\leq&\frac{1}{h}\mathbb{E}\left[\int_0^{\tau^{c^\varepsilon}\wedge h}[({\cal{L}}\varphi)(X^{\hat{x},\pi^{\varepsilon},{c}^{\varepsilon}}_s,
                                     {\pi}^{\varepsilon}_s, c^{\varepsilon}_s)+f(c^{\varepsilon}_s,Y^{c^{\varepsilon},J^N}_s)-\delta \nu Y^{c^{\varepsilon},J^N}_s]ds\right]\\
                                     &=&\frac{1}{h}\mathbb{E}\left[\int_0^{\tau^{c^\varepsilon}\wedge h}[({\cal{L}}\varphi)(\hat{x},
                                     {\pi}^{\varepsilon}_s, c^{\varepsilon}_s)+f(c^{\varepsilon}_s,J^N(\hat{x}))-\delta \nu J^N(\hat{x})]ds\right]\\
                                     &&+\frac{1}{h}\mathbb{E}\bigg[\int_0^{\tau^{c^\varepsilon}\wedge h}[({\cal{L}}\varphi)(X^{\hat{x},\pi^{\varepsilon},{c}^{\varepsilon}}_s,
                                     {\pi}^{\varepsilon}_s, c^{\varepsilon}_s)+f(c^{\varepsilon}_s,Y^{c^{\varepsilon},J^N}_s)
                                      -\delta \nu Y^{c^{\varepsilon},J^N}_s]\\
                                      && \ \ -[({\cal{L}}\varphi)(\hat{x},
                                     {\pi}^{\varepsilon}_s, c^{\varepsilon}_s)+f(c^{\varepsilon}_s,J^N(\hat{x}))-\delta \nu J^N(\hat{x})]ds\bigg]\\
                           &\leq& [-\delta \nu J^N(\hat{x})+{\mathbf{H}}_N(\hat{x},J^N(\hat{x}),\varphi_x(\hat{x}),\varphi_{xx}(\hat{x}))]
                           \left[1-\frac{2C_1h^{\frac{1}{2}}}{\hat{x}}\right]+I,
\end{eqnarray*}
where, from $\varphi\in C_b^2([\frac{\hat{x}}{2},\infty))$, (\ref{conti,X}) and (\ref{conti,J_N}),
\begin{eqnarray*}
|I|&\leq& \frac{1}{h}\mathbb{E}\bigg[\int_0^{h}|({\cal{L}}\varphi)(X^{\hat{x},\pi^{\varepsilon},{c}^{\varepsilon}}_{\tau^{c^\varepsilon}\wedge s},
                                     {\pi}^{\varepsilon}_{\tau^{c^\varepsilon}\wedge s}, c^{\varepsilon}_{\tau^{c^\varepsilon}\wedge s})+f(c^{\varepsilon}_{\tau^{c^\varepsilon}\wedge s},Y^{c^{\varepsilon},J^N}_{\tau^{c^\varepsilon}\wedge s})-\delta \nu Y^{c^{\varepsilon},J^N}_{\tau^{c^\varepsilon}\wedge s}]\\
                                     &&-[({\cal{L}}\varphi)(\hat{x},
                                     {\pi}^{\varepsilon}_{\tau^{c^\varepsilon}\wedge,s}, c^{\varepsilon}_{\tau^{c^\varepsilon}\wedge s})+f(c^{\varepsilon}_{\tau^{c^\varepsilon}\wedge s},J^N(\hat{x}))-\delta \nu J^N(\hat{x})|ds\bigg]\rightarrow0
                                     \ \mbox{as}\ h\rightarrow0^+.
\end{eqnarray*}
Letting $h\rightarrow0$ and $\varepsilon\rightarrow0$, we have that
$$-\delta\nu J^N(\hat{x})+{\mathbf{H}}_N(\hat{x},J^N(\hat{x}),\varphi_{x}(\hat{x}),\varphi_{xx}(\hat{x}))\geq0.$$
Next, let  $\varphi\in {\cal{A}}^-(\hat{x},J^N)$
                  with
                   $\hat{x}\in \mathbb{R}_{++}$,  for any given $(\bar{\pi},\bar{c})$ with {$|\bar{\pi}|\leq g(\hat{x})$ }and $\bar{c}\leq N$, let $(\pi,c)\in\mathcal{A}^N(\hat{x})$
                   such that $(\pi_0,c_0)=(\bar{\pi},\bar{c})$ and are continuous on $t=0$. Then by the DPP ( Proposition \ref{pro: dpp}), we obtain: for any $h>0$,
 \begin{eqnarray*}\label{4.9}
                            0=J^N(\hat{x})+\varphi (\hat{x})
                          &\geq& \mathbb{E}\left[J^N(X^{\hat{x},\pi,c}_{\tau^{c}\wedge h})
                           +\int_0^{\tau^{c}\wedge h}[f(c_s,Y^{c,J^N}_s)-\delta \nu Y^{c,J^N}_s]ds\right]
                          +\varphi (\hat{x})\\
                           &\geq& \mathbb{E}\left[-\varphi(X^{\hat{x},\pi,c}_{\tau^{c}\wedge h})
                           +\int_0^{\tau^{c}\wedge h}[f(c_s,Y^{c,J^N}_s)-\delta \nu Y^{c,J^N}_s]ds\right]
                          +\varphi (\hat{x}).
\end{eqnarray*}
 Applying  It\^{o} formula  to
                         ${\varphi}(X^{\hat{x},\pi,{c}}_s)$,   we get that
\begin{eqnarray*}
                       0
                           &\geq&\frac{1}{h}\mathbb{E}\left[\int_0^{\tau^{c}\wedge h}[-({\cal{L}}\varphi)(X^{\hat{x},\pi,{c}}_s,
                                     {\pi}_s, c_s)+f(c_s,Y^{c,J^N}_s)-\delta \nu Y^{c,J^N}_s]ds\right].
\end{eqnarray*}
Let $h\rightarrow0$  we have that
$$-({\cal{L}}\varphi)(\hat{x},\bar{\pi}, \bar{c})+f(\bar{c},J^N(\hat{x}))-\delta \nu J^N(\hat{x})\leq0.$$
By the arbitrariness of $(\bar{\pi},\bar{c})$,
$$-\delta\nu J^N(\hat{x})+{\mathbf{H}}_N(\hat{x},J^N(\hat{x}),-\varphi_{x}(\hat{x}),-\varphi_{xx}(\hat{x}))\leq0.$$
\end{proof}
The second step is to show that $J^{N}(x)$ converges to $J(x)$ as $N \rightarrow \infty$, for all $x > 0$.

\begin{proposition}
    We have
    \begin{eqnarray}\label{appN}
J^N(x)\rightarrow J(x)\ \mbox{as}\ N\rightarrow\infty.
\end{eqnarray}
\end{proposition}
\begin{proof}
By the definition, it is clear that $\limsup_{N\rightarrow \infty}J^N(x)\leq J(x)$. It suffices to prove that 
\begin{eqnarray}\label{JNconvJ}
\liminf_{N\rightarrow \infty}J^N(x)\geq J(x).
\end{eqnarray}
For any $\epsilon>0$, by the definition of $J(x)$, there exists $c^{\epsilon}\in \mathcal{C}(x)$ such that 
$$
         J(x)\leq V^{c^{\epsilon}}(x)+\epsilon.
$$
By Theorems 6.5 and 6.7 in \cite{HHJ23b}, we have, for every $c\in \mathcal{C}(x)$,
\begin{eqnarray}\label{VNconve}
V^{c\wedge N}_t(x)\rightarrow V^c_t(x), \ \mbox{as}\ N\rightarrow\infty.
\end{eqnarray}
By (\ref{VNconve}), there exists $c^{\epsilon,N}$ such that $c^{\epsilon,N}\in \mathcal{C}^N(x)$ and 
$$
          V^{c^{\epsilon,N}}(x)\geq V^{c^{\epsilon}}(x)-\epsilon.
$$
Therefore,
$$
         J(x)\leq V^{c^{\epsilon,N}}(x)+2\epsilon\leq J^N(x)+2\epsilon.
$$
Letting $N\rightarrow \infty$ and $\epsilon\rightarrow 0$, we get (\ref{JNconvJ}).
\end{proof}
Finally, the following theorem establishes the existence of a viscosity solution to the HJB equation (\ref{eq:HJB}).

\begin{theorem}\label{existence}
 The optimal value
                          function $J$ defined by (\ref{eq:ez-problem}) is a
                          viscosity solution to  HJB equation (\ref{eq:HJB}).

\end{theorem}
\begin{proof}
    Let   $\varphi\in \mathcal{A}^+(\hat{x}, J)$ with
  $\hat{x}>0$.  Then, for any $N > 0$, 
  $$J^N(x)-\varphi(x)\leq J(x)-\varphi(x)\leq J(\hat{x})-\varphi(\hat{x}) =0.$$
 Denote $\varphi_{1}(x):=\varphi(x)+|x-\hat{x}|^2$. Then, there exists $x_N$ such that
 $$
         J^N(x_N)-\varphi_1(x_N)=\sup_{x\in \mathbb{R}_+}[J^N(x)-\varphi_1(x)].
 $$
We claim that
\begin{equation}
\label{eq:gamma}
|x_N-\hat{x}|\rightarrow0  \ \ \mbox{as} \ \ N\rightarrow\infty.
\end{equation}
 Indeed, suppose not. Then we can assume that there exists a constant  $\nu_0>0$
 such that
$$
              |x_N-\hat{x}|^2  \geq\nu_0.
$$
It then follows that
\begin{eqnarray*}
   &&0=(J- {{\varphi}})(\hat{x})= \lim_{N\rightarrow\infty}(J^N-{{\varphi_{1}}})(\hat{x})
   \leq \limsup_{N\rightarrow\infty}[(J^N-\varphi)(x_N)-|x_N-\hat{x}|^2]\\
   &\leq&\limsup_{N\rightarrow\infty}[(J-\varphi)(x_N)+(J^N-J)(x_N)]
     -\nu_0\leq (J- {{\varphi}})(\hat{x})-\nu_0=-\nu_0,
\end{eqnarray*}
 contradicting $\nu_0>0$. We have thus proved (\ref{eq:gamma}).  Let $N$ be large enough, by $\hat{x}>0$, we have
 $
 x_N>0
 $.
 
 By Theorem \ref{existenceN}, $J^N$ is a viscosity subsolution of HJB equation (\ref{eq:HJBN}). Then, we have
$$
                         -\delta\nu J^N(x_N)+\mathbf{H}_N(x_N,J_N(x_N),(\varphi_1)_x(x_N),(\varphi_1)_{xx}(x_N))\geq0.
$$
Letting $N\rightarrow+\infty$,
\begin{align*}
                         -\delta\nu J(\hat{x})+\mathbf{H}(\hat{x},\varphi(\hat{x}),\varphi_x(\hat{x}),\varphi_{xx}(\hat{x}))\geq0.
\end{align*}
Next,  let   $\varphi\in \mathcal{A}^-(\hat{x}, J)$ with
  $\hat{x}>0$ and denote $\varphi_{2}(x):=\varphi(x)+|x-\hat{x}|^2$. Then there exist $y_N$ such that
 $$
         J^N(y_N)+\varphi_2(y_N)=\inf_{x\in \mathbb{R}_+}[J^N(x)+\varphi_2(x)].
 $$
We can prove that
\begin{align*}
|y_N-\hat{x}|\rightarrow0  \ \ \mbox{as} \ \ N\rightarrow\infty.
\end{align*}
 Otherwise, there exists a constant  $\nu_1>0$
 such that
$$
              |y_N-\hat{x}|  \geq\nu_1.
$$
Then, we obtain that
\begin{eqnarray*}
   &&0=(J+ {{\varphi}})(\hat{x})= \lim_{N\rightarrow\infty}(J^N+{{\varphi_{2}}})(\hat{x})
   \geq \liminf_{N\rightarrow\infty}[(J^N+\varphi)(y_N)+|y_N-\hat{x}|^2]\\
   &\geq&\liminf_{N\rightarrow\infty}[(J+\varphi)(y_N)+(J^N-J)(y_N)]
     +\nu_1\geq (J+{{\varphi}})(\hat{x})+\nu_1=\nu_1,
\end{eqnarray*}
 contradicting $\nu_1>0$.   Let $N$ be large enough, by $\hat{x}>0$, we have
 $
 y_N>0
 $. Again, since $J^N$ is a viscosity supersolution of HJB equation (\ref{eq:HJBN}), we have
$$
                         -\delta\nu J^N(y_N)+\mathbf{H}_N(y_N,J^N(y_N),-(\varphi_1)_{x}(y_N),-(\varphi_1)_{xx}(y_N))\leq0.
$$
Letting $N\rightarrow\infty$,
\begin{align*}
                         -\delta\nu J(\hat{x})+\mathbf{H}(\hat{x},J(\hat{x}),-\varphi_{x}(\hat{x}),-\varphi_{xx}(\hat{x}))\leq0.
\end{align*}
We have completed the proof. 
\end{proof}
\begin{rem}\label{rem:r>1}
    We can not directly prove Theorem \ref{existence} due to the lack of the  DPP 
    for 
    the optimal value function $J(x)$. Instead, we first establish Theorem \ref{existenceN} by applying the DPP to $J^N(x)$, and then derive  Theorem \ref{existence} using (\ref{appN}) together with the fact that $\mathbf{H}_N\rightarrow \mathbf{H}$ as $N\rightarrow \infty$. Several important studies have addressed the DPP for Epstein–Zin preferences in a discrete-time setting; see \cite{BJ18}, \cite{S22}, and \cite{SWZ24}. By contrast, in the continuous-time setting, it is not known ex ante whether a solution to the corresponding HJB equation exists. A key result  is to establish the existence of a viscosity solution, while the DPP plays a central role in the proof.
\end{rem}

\begin{rem}
   Theorem \ref{existence} also holds for the case of $R>1$, upon a careful review of the proof procedures above. For example, for $R>1$, \eqref{eq:ez-bound} and Proposition \ref{pro:ez-concave} also hold. Moreover, Lemma \ref{lem:bsde1}, Lemma \ref{lem:bsde-rs<1} and DPP can be verified in a similar manner. See Appendix D for details. 
\end{rem}

\subsection{Uniqueness of viscosity solutions}
In this subsection, we characterize the value function as the unique viscosity solution of the corresponding HJB equation (\ref{eq:HJB}).
Its proof depends on the comparison principle established by the following result.

\begin{theorem}\label{12052} 
If $u$ is a concave  viscosity subsolution of (\ref{eq:HJB}) on $\mathbb{R}_+$ with sublinear growth and $u(0)\leq0$ and $v$ with $v(0)\geq0$ is a bounded from below, strictly concave
 supresolution of (\ref{eq:HJB}) on  ${\Omega}$, the $u\leq v$ on $\mathbb{R}_+$.
\end{theorem}
\begin{proof}
    We proceed by contradiction. Suppose, for the sake of argument, that
    \begin{align*}
           \sup_{x\in \mathbb{R}_+}[u(x)-v(x)]>0.
    \end{align*}
    By the continuity of $u$ and $v$, there exists an $\hat{x}>0$ such that
    \begin{align*}
         2\tilde{m}:=  u(\hat{x})-v(\hat{x})>0.
    \end{align*}
    Then, for sufficiently small $\theta>0$,
    \begin{align*}
           \sup_{x\in \mathbb{R}_+}[u(x)-v(x)-\theta x]
          \geq\tilde{m}>0.
    \end{align*}
    Since $u$  is concave with sublinear growth and $v$ is bounded from below, there exists $\bar{x}\in \mathbb{R}_{++}$ such that
    \begin{align*}
           \sup_{x\in \mathbb{R}_+}[u(x)-v(x)-\theta x]
          =u(\bar{x})-v(\bar{x})-\theta \bar{x}
          \geq \tilde{m}>0.
    \end{align*}

   To proceed, we present the following lemma, with its proof provided in Appendix A.
 \begin{lemma}
    \label{lemma:unqueness}
        For $\lambda>0$, we define $\varphi:\mathbb{R}_+\times \mathbb{R}_+\rightarrow \mathbb{R}$ by
    $$
     \varphi(x,y)=u(x)-v(y)-\lambda|x-y|^2-\frac{\theta}{2} (x+y).
    $$
    For each fixed $\lambda$ and $\theta$, the function $\varphi(x,y)$ attains its maximum at a point  $(x_0,y_0):=(x_0(\theta,\lambda), y_0(\theta,\lambda))$ such that
    \begin{eqnarray}\label{12055}
        \lim_{\lambda\rightarrow+\infty}\lambda |x_0-y_0|^2=0.
    \end{eqnarray}
    Moreover, there exists $\bar{x}_0:=\bar{x}_0(\theta)>0$  such that 
    \begin{eqnarray}\label{101025}
    \lim_{\lambda\rightarrow+\infty}x_0(\theta,\lambda)=\bar{x}_0(\theta), 
    \quad
    \lim_{\theta\rightarrow0^+}\theta [g(\bar{x}_0)\vee \bar{x}_0]=0, \quad \mbox{and}
    \end{eqnarray}
    \begin{eqnarray}\label{0318zhou}
               u(\bar{x}_0)-v(\bar{x}_0)\geq
               \sup_{x\in \mathbb{R}_+}[u(x)-v(x)-\theta x]\geq\tilde{m}>0
               \geq u(0)-v(0).
\end{eqnarray}
 Here and in the following, the limit $\lambda\rightarrow+\infty$ is taken along subsequences which, to simplify notation, we denote the same way as the whole family. 
    \end{lemma}

 {\em  We now return to the proof of Theorem \ref{12052}.}
  
    By (\ref{12055}) and (\ref{101025}), for each fixed   sufficiently small $\theta>0$, there exists a constant $\Delta_\theta>0$ large enough that
    \begin{eqnarray}\label{x_0y_0>0}
              x_0,y_0\geq \frac{\bar{x}_0}{2}>0, \ \ \mbox{for all} \ \lambda\geq\Delta_\theta.
    \end{eqnarray}
\par
We put, for $x,y\in \mathbb{R}_+$,
\begin{align*}
                u_1(x)=u(x)-\frac{\theta}{2} x, ~v_1(y)=v(y)+\frac{\theta}{2} y \text{~~and~~} \psi(x,y)=\lambda|x-y|^2.
\end{align*}
From above all, $u_1(x)-v_1(y)-\psi(x,y)$ has  a maximum at $(x_0,y_0)\in \mathbb{R}_{++}\times \mathbb{R}_{++}$ over $\mathbb{R}_+\times \mathbb{R}_+$. Then by Theorem 3.2 in \cite{c92} and Lemma 6.7 in \cite{zhou24}, there exist $X,Y\in \mathbb{R}$, sequences $x_k,y_k\in \mathbb{R}_{++}$ and the sequences of functionals $\varphi_k,\psi_k\in C^2(\mathbb{R}_+)$ bounded from below such that
$$u_1(x)-\varphi_k(x)\ \ (\mbox{resp.}, v_1(x)+\psi_k(x))$$
has a strict maximum (resp. minimum) $0$ at $x_k \ (\mbox{resp.}, y_k )$ over $\mathbb{R}_+$, and
$$
(x_k,u_1(x_k),(\varphi_k)_x(x_k),(\varphi_k)_{xx}(x_k))\underrightarrow{k\rightarrow+\infty}(x_0,u_1(x_0),2\lambda(x_0-y_0), X),
$$
$$
(y_k,v_1(y_k),(\psi_k)_x(y_k),(\psi_k)_{xx}(y_k))\underrightarrow{k\rightarrow+\infty}(y_0,v_1(x_0),2\lambda(y_0-x_0), Y),
$$
 and $X,Y$ satisfy the following inequality
 \begin{align*}
{-6\lambda}\left(\begin{array}{cc}
                                    I & 0\\
                                    0 & I
                                    \end{array}\right)\leq \left(\begin{array}{cc}
                                    X&0\\
                                    0&Y
                                    \end{array}\right)\leq  6\lambda\left(\begin{array}{cc}
                                    I&-I\\
                                    -I&I
                                    \end{array}\right).
\end{align*}
Let $$
           \hat{\varphi}_k(x)
           :=\frac{\theta}{2} x
                +\varphi_k(x)
$$
and
$$
           \hat{\psi}_k(y)
           :=\frac{\theta}{2} y
                +\psi_k(y).
$$
It is clear that $\hat{\varphi}_k,\hat{\psi}_k\in C^2(\mathbb{R}_+)$ and
$$
 (u-\hat{\varphi}_k)(x_k)=\sup_{x\in \mathbb{R}_+}[u-\hat{\varphi}_k](x),
 \ \ \
  (v+\hat{\psi}_k)(y_k)=\inf_{y\in \mathbb{R}_+}[v+\hat{\psi}_k](y).
$$
Now, for each sufficiently small  $\theta>0$ and  $\lambda>\Delta_\theta$, it follows from the definition of viscosity solutions that
\begin{eqnarray}\label{12061}
\begin{aligned}
        -\delta \nu u(x_k)&+\sup_{{|\pi|\leq g(x_k)}}\left[\pi(\mu-r)(\hat{\varphi}_k)_{x}(x_k)+{1\over 2}\sigma^2\pi^2 (\hat{\varphi}_k)_{xx}(x_k)\right]\\
        &+\sup_{c\geq 0}\left[f(c,u(x_k))-c(\hat{\varphi}_k)_{x}(x_k)\right]+rx_k(\hat{\varphi}_k)_{x}(x_k)\geq0,\qquad \mbox{and}
        \end{aligned}
\end{eqnarray}
\begin{eqnarray}\label{12062}
\begin{aligned}
        -\delta\nu v(y_k)&+\sup_{{|\pi|\leq g(y_k)}}\left[-\pi(\mu-r)(\hat{\psi}_k)_{x}(y_k)-{1\over 2}\sigma^2\pi^2 (\hat{\psi}_k)_{xx}(y_k)\right]\\
        &+\sup_{c\geq 0}\left[f(c,v(y_k))+c(\hat{\psi}_k)_{x}(y_k)\right]-ry_k(\hat{\psi}_k)_{x}(y_k)\leq0.
        \end{aligned}
\end{eqnarray}
Notice that, by (\ref{12062}),
$$
            0>\lambda(y_0-x_0)+\frac{\theta}{2}=\lim_{k\rightarrow+\infty}(\hat{\psi}_k)_{x}(y_k).
$$
Then,
letting $k\rightarrow+\infty$ in (\ref{12061}) and (\ref{12062}),
\begin{eqnarray}\label{12063}
\begin{aligned}
        -\delta&\nu u(x_0)+\sup_{{|\pi|\leq g(x_0)}}\left[\pi(\mu-r)\Big(\lambda(x_0-y_0)+\frac{\theta}{2}\Big)+{1\over 2}\sigma^2\pi^2 X\right]\\
        &+\sup_{c\geq 0}\left[f(c,u(x_0))-c\Big(\lambda(x_0-y_0)+\frac{\theta}{2}\Big)\right]+r x_0\left[\lambda(x_0-y_0)+\frac{\theta}{2}\right]\geq0,
        \end{aligned}
\end{eqnarray}
and
\begin{eqnarray}\label{12064}
\begin{aligned}
        -\delta\nu &v(y_0)+\sup_{{|\pi|\leq g(y_0)}}\left[-\pi(\mu-r)\Big(\lambda(y_0-x_0)+\frac{\theta}{2}\Big)-{1\over 2}\sigma^2\pi^2 Y\right]\\
        &+\sup_{c\geq 0}\left[f(c,v(y_0))+c\Big(\lambda(y_0-x_0)+\frac{\theta}{2}\Big)\right]-ry_0\left[\lambda(y_0-x_0)+\frac{\theta}{2}\right]\leq0.
        \end{aligned}
\end{eqnarray}
Combining (\ref{12063}) and (\ref{12064}),
\begin{eqnarray}\label{12065}
\begin{aligned}
        0\leq&-\delta\nu (u(x_0)-v(y_0))+\sup_{{|\pi|\leq g(x_0)}}\left[\pi(\mu-r)\Big(\lambda(x_0-y_0)+\frac{\theta}{2}\Big)+{1\over 2}\sigma^2\pi^2 X\right]\\
        &-\sup_{{|\pi|\leq g(y_0)}}\left[-\pi(\mu-r)\Big(\lambda(y_0-x_0)+\frac{\theta}{2}\Big)-{1\over 2}\sigma^2\pi^2 Y\right]\\
        &+\sup_{c\geq 0}\left[f(c,u(x_0))-c\Big(\lambda(x_0-y_0)+\frac{\theta}{2}\Big)\right]+r x_0\left[\lambda(x_0-y_0)+\frac{\theta}{2}\right]\\
       & -\sup_{c\geq 0}\left[f(c,v(y_0))+c\Big(\lambda(y_0-x_0)+\frac{\theta}{2}\Big)\right]+ry_0\left[\lambda(y_0-x_0)+\frac{\theta}{2}\right]\\
        \leq&-\delta\nu (u(x_0)-v(y_0))+\sup_{{|\pi|\leq g(x_0)}}\left[\pi(\mu-r){\theta}+{1\over 2}\sigma^2\pi^2 (X+Y)\right]\\
        &+\sup_{c\geq 0}\left[f(c,u(x_0))-f(c,v(y_0))-c\theta\right]+r\lambda|x_0-y_0|^2+\frac{r\theta}{2} (x_0+y_0)+I,
       \end{aligned}
\end{eqnarray}
where
\begin{eqnarray*}
I:=&\sup_{{|\pi|\leq g(x_0)}}\left[-\pi(\mu-r)\big(\lambda(y_0-x_0)+\frac{\theta}{2}\big)-{1\over 2}\sigma^2\pi^2 Y\right]\\
&-\sup_{{|\pi|\leq g(y_0)}}\left[-\pi(\mu-r)\big(\lambda(y_0-x_0)+\frac{\theta}{2}\big)-{1\over 2}\sigma^2\pi^2 Y\right].
\end{eqnarray*}
We claim that 
\begin{eqnarray}\label{101125}
\lim_{\lambda\rightarrow+\infty}I=0,
\end{eqnarray}
and its proof provided in Appendix A.
Noting that 
$X+Y\leq0$, letting $\lambda\rightarrow+\infty$ in (\ref{12065}), by (\ref{12055}), (\ref{101025}) and (\ref{101125}),
\begin{eqnarray}\label{03182zhou}
        0
        \leq&-&\delta\nu (u(\bar{x}_0)-v(\bar{x}_0))+g(\bar{x}_0)(\mu-r)\theta\\
        &+&\sup_{c\geq 0}[f(c,u(\bar{x}_0))-f(c,v(\bar{x}_0))-c\theta]+{r\theta}\bar{x}_0.\nonumber
\end{eqnarray}
Since $\delta\nu>0$ and $\rho<0$, by (\ref{0318zhou}) and the definition of $f$,
$$
-\delta\nu (u(\bar{x}_0)-v(\bar{x}_0))\leq -\delta\nu\tilde{m}, \ \
f(c,u(\bar{x}_0))-f(c,v(\bar{x}_0))\leq0.
$$
 Then, letting $\theta\rightarrow0$ in (\ref{03182zhou}), by (\ref{101025}) the following contradiction is induced:
 \begin{align*}
        0\leq-\delta\nu\tilde{m}<0.
\end{align*}The proof is completed.
\end{proof}
\begin{rem}\label{rem:r>1 uni}
   Since $0 < R < 1$, the HJB equation (\ref{eq:HJB}) is defined on $\mathbb{R}_+$, while classical HJB equation is defined on $\mathbb{R}$ or an open subset of $\mathbb{R}$. Hence, we can not apply the result of classical HJB equation to our case. In particular, we need to prove (\ref{x_0y_0>0}) holds true. Specifically, 
      to obtain the maximum point for the auxiliary function $\varphi$, we add the term $\frac{\theta}{2}(x+y)$ to $\varphi$. Since $\delta\nu$ is fixed, it remains to prove that (\ref{03181zhou}) holds. This constitutes the key difference from the viscosity solution theory of the classical HJB equation as in \cite{FS06}.
\end{rem}
\begin{rem}
Compared to the classical HJB equations, another key distinction is that the supremum in the HJB equation  (\ref{eq:HJB}) depends on $g(x)$. This presents significant challenges. In fact, we need to establish the validity of equality (\ref{101125}), which requires the use of the strict concavity of supresolution $v$.
\end{rem}
    \begin{rem}
While the existence of a viscosity solution can be established for any $R \ne 1$, the uniqueness for  $R > 1$ remains unclear. In this case, the optimal value function does not satisfy a sublinear growth condition, and the boundary condition $J(0)$ is not well-defined. More importantly, Lemma \ref{lemma:unqueness} does not apply, as the function $\varphi(x,y)$ does work for the situation $R>1$.
Consequently, we cannot establish the uniqueness of the viscosity solution. In fact, this difficulty arises even under the standard CRRA utility when $R > 1$, as shown in \cite{Z94}.
\end{rem}
\begin{proof}[Proof of Theorem \ref{th-unique}]

From Theorems (\ref{existence}) and  (\ref{12052}), the value function $J$ defined by (\ref{eq:ez-problem}) is the unique viscosity solution of  HJB equation (\ref{eq:HJB}) in the class of bounded from below, strictly concave functions on $\mathbb{R}_+$. 
\end{proof}

\section{The Proof of Theorem \ref{th-c2}}\label{sec:smoothness}
In this section, we establish the smoothness of the optimal value function. Although our approach is motivated by the classical time-separable case in \cite{Z94}, we must address the challenges introduced by the Epstein–Zin aggregator.



We consider a product probability space on $(\bar{\Omega},\bar{\mathcal{F}},\bar{\mathbb{P}}):=(\Omega\times\tilde{\Omega}, \mathcal{F}\times\tilde{\mathcal{F}},\mathbb{P}\times\tilde{\mathbb{P}})$, on which $(B,\tilde{B})$ is a two-dimensional Brownian motion. We also require that $\{\tilde{B}_t\}_{t\geq 0}$ is a Brownian motion independent of $\{B_t\}_{t\geq 0}$, which is defined on $(\tilde{\Omega},\tilde{\mathcal{F}},\tilde{\mathbb{P}})$.

For any  $\epsilon>0$, we consider an auxiliary problem where the wealth process evolves as
\begin{equation}
\label{eq:wealth-epsilon}
dX_t^{\epsilon}=[rX^{\epsilon}_t+(\mu-r)\pi_t^{\epsilon}-c_t^{\epsilon}]dt+\pi^{\epsilon}_t \sigma dB_t+\sigma \epsilon X_t^{\epsilon}d\tilde{B}_t, ~~ X_0^{\epsilon}=x>0.  
\end{equation}
Compared to the wealth equation (\ref{eq:wealth}), there exists an extra diffusion term $\sigma \epsilon X_t^{\epsilon}d\tilde{B}_t$.
Furthermore, we define the value function by
\begin{equation}
\label{eq:ez-epsilon}
J^{\epsilon}(x)= \sup_{(\pi^{\epsilon},c^{\epsilon})\in {\cal{A^{\epsilon}}}(x)}V_0^{c^\epsilon}=\sup_{(\pi^{\epsilon},c^{\epsilon})\in {\cal{A^{\epsilon}}}(x)}
\bar{\mathbb{E}}\left[\int_{0}^{\infty}e^{-\delta s} f(c^{\epsilon}_{s},V_{s}^{c^{\epsilon}})ds\right],~~~~x>0,
\end{equation}
where $\mathcal{A}^{\epsilon}(x)$ represents that the admissible strategy $(\pi^\epsilon,c^\epsilon)$, $\bar{\mathcal{F}}_t$-progressive measurable processes, such that
$c^\epsilon_t\geq0$, $X_t^\epsilon\geq0$,  $|\pi^\epsilon_t|\leq g(X_t^\epsilon)$, and $\int_0^t c_u^\epsilon+(\pi_u^\epsilon)^2du<+\infty$, $\bar{P}$-a.s., for all $t\geq0$. 

The proof of the following proposition is similar to the proof of Proposition \ref{pro:basic pro} and Proposition \ref{pro:ez-concave}.

\begin{proposition}
  The function $J^\epsilon$ satisfies sublinear growth condition, i.e., $\forall x>0$, 
\begin{align}\label{eq:v-epsilon-bounded}
\left(\delta-\frac{1}{2\nu}\sigma^2\epsilon^2R(1-R)\right)^{-\nu}\frac{(rx)^{1-R}}{1-R}\leq J^{\epsilon}(x)\leq (\eta+\frac{R(1-S)}{2S}\sigma^2\epsilon^2)^{-\nu S}\frac{x^{1-R}}{1-R}.
\end{align}Besides, $J^\epsilon$ is also strictly increasing and strictly concave  on $\mathbb{R}_{++}$. 
\end{proposition}
\begin{proof}
In the auxiliary problem, for sufficiently small $\epsilon$ such that $\delta \nu>\frac{1}{2}\epsilon^2R(1-R)$, we can choose $(\pi^{\epsilon},c^{\epsilon})=(0,rX_t^{\epsilon})\in\mathcal{A}^{\epsilon}(x)$. In this case, $X_t^\epsilon=x\exp{(\sigma \epsilon \tilde{B}_t-\frac{1}{2}\sigma^2\epsilon^2 t)}>0$ for $x>0$. To prove the lower bound, we conjecture that $$V_t^{c^{\epsilon}}=Ae^{-\beta t}\frac{(X_t^{\epsilon})^{1-R}}{1-R}, ~~t\geq0$$ is the solution to $$V_t^ {c^{\epsilon}}=\bar{\mathbb{E}}\left[\int_{t}^{\infty}e^{-\delta s} f(c^{\epsilon}_{s},V_{s}^{c^{\epsilon}})ds~|~\mathcal{F}_t~\right], ~~t\geq0. $$ Substituting the explicit form of $X_t^\epsilon$ into it yields $\beta=\delta\nu$ and 
$$A=\left(\delta-\frac{1}{2\nu}\sigma^2\epsilon^2R(1-R)\right)^{-\nu}r^{1-R}.$$ The concavity is similar to Proposition \ref{pro:ez-concave}. Moreover, the fact that  $J^\epsilon$  admits the above lower bound implies that it is strictly increasing and strictly concave. The upper bound of \eqref{eq:v-epsilon-bounded} follows from an argument analogous to that of Theorem 8.1 in \cite{HHJ23b}.
\end{proof}


\begin{proposition}
\label{prop:smooth-epsilon}
  The function $J^{\epsilon}$ is the unique viscosity solution of the following HJB equation, i.e., for all $x>0$, 
\begin{align}
\label{eq:HJB-epsilon}
\delta \nu J^{\epsilon}(x)=&\sup_{|\pi|\leq g(x)}[\pi(\mu-r)J_x^{\epsilon}+{1\over 2}\sigma^2(\pi^2+\epsilon^2x^2) J_{xx}^{\epsilon}]+\sup_{c\geq 0}[f(c,J^{\epsilon})-cJ_x^{\epsilon}]+rxJ_x^{\epsilon}.
\end{align}
Moreover, $J^{\epsilon}$ is also the unique $C^2$ smooth solution of \eqref{eq:HJB-epsilon}.
\end{proposition}


\begin{proof}
In the first part, the proof follows a similar argument to Theorem \ref{th-unique}, showing that $J^{\epsilon}$ is the unique viscosity solution to the HJB equation \eqref{eq:HJB-epsilon}. We omit the details here. For the second part, we note that equation \eqref{eq:HJB-epsilon} satisfies the uniform ellipticity condition after adding $\frac{1}{2}\sigma^2 \epsilon^2 x^2$ in the local extension. Therefore, by \cite{K87}, $J^{\epsilon}$ is also the unique $C^2$ smooth solution of \eqref{eq:HJB-epsilon}.
\end{proof}

The following proposition is rather technical and essential to the uniqueness theorem.

\begin{proposition}
\label{prop:smooth}
For any interval $[x_1,x_2]\subset (0,+\infty)$, there exists positive constants $K_1,K_2,K_3$ (not depending on $\epsilon$), such that for all $x\in [x_1,x_2]$ and  $\epsilon>0$, we have
 \begin{eqnarray}\label{eq:bound sec der}
 |J_{xx}^{\epsilon}(x)|\leq K_1, ~~
 J(x)\in[K_2,K_3],~~
 J_x(x)\in[K_2,K_3]. 
 \end{eqnarray}
\end{proposition}
\proof 
For later purposes, we define an auxiliary function $\xi$ and two more intervals $[y_1,y_2],[z_1,z_2]$ such that $[x_1,x_2]\subset [y_1,y_2]\subset [z_1,z_2]\subset(0,+\infty)$.
 We assume $\xi:\mathbb{R}_+\rightarrow \mathbb{R}_+$ satisfies
 \begin{itemize}
     \item[(i)] $\xi\in C_0^{\infty}$;
     \item[(ii)] $\xi \equiv 1$ on $[x_1,x_2]$, $\xi\equiv 0$ on $[y_1,y_2]^c$ and $0\leq \xi\leq 1$ otherwise;
     \item[(iii)] $|\xi_x|\leq M\xi^p$,  $|\xi_{xx}|\leq M\xi^p$ with some $p\in(0,1)$ and $M>0$.
 \end{itemize}

For the sake of simplicity, we suppress the $\epsilon$ notation until the end of this proposition.   Define a function $Z:[z_1,z_2]\rightarrow \mathbb{R}$ by \begin{equation*}
Z(x)=\xi^2(x)J^2_{xx}(x)+\lambda_1 J_x^2(x)-\lambda_2 J(x), ~~x\in[z_1,z_2],
\end{equation*}
where $\lambda_1,\lambda_2>0$ to be determined later. Furthermore,
by using the estimations \eqref{eq:v-epsilon-bounded} and the similar technique of Proposition 2 in \cite{XY16}, it implies that there exist two positive constants $K_2$ and $K_3$ (not depending on $\epsilon$, only on $[z_1,z_2]$) such that
\begin{equation}
\label{eq: bound for V and V_x}
J_x(x)\in[K_2,K_3] \text{~~~and~~~} J(x)\in[K_2,K_3], ~~~~\forall x\in [x_1,x_2].
\end{equation}

Next, we split into several cases with respect to the maximum of $Z$ on $[z_1, z_2]$:

{\bf Case 1:} $Z$ achieves its maximum at $x_0\notin \text{supp}\ \xi$, then we have
\begin{equation*}
J^2_{xx}(x)+\lambda_1 J_x^2(x)-\lambda_2 J(x) \leq Z(x_0)= \lambda_1 J_x^2(x_0)-\lambda_2 J(x_0),  ~~\forall x\in[z_1,z_2].
\end{equation*}
Since $J$ and $J_x$ are both positive, we therefore have
\begin{equation*}
J_{xx}^2(x)\leq \lambda_1J^2_x(x_0)+\lambda_2 J(x), ~~\forall x\in[x_1,x_2].
\end{equation*}
By utilizing (\ref{eq: bound for V and V_x}), we have reached (\ref{eq:bound sec der}).

{\bf Case 2:} The function $Z$ reaches its maximum at $x_0\in \text{supp}~\xi\subset[y_1, y_2]$.
Since $J_x>0$ and $J_{xx}<0$, it is clear that $-{\mu-r\over \sigma^2}{J_x(x)\over J_{xx}(x)}>0$.  Therefore, $|-{\mu-r\over \sigma^2}{J_x(x)\over J_{xx}(x)}|={\mu-r\over \sigma^2}{J_x(x)\over J_{xx}(x)}$ and we only need to consider the following two cases:

{\bf Case 2 (a):} The case
\begin{equation*}
x_0 \in \left\{x\in[y_1,y_2]:-{\mu-r\over \sigma^2}{J_x(x)\over J_{xx}(x)}\geq g(x)\right\}.
\end{equation*}
In this situation, by the increasing property of $g$, it implies that
$$|J_{xx}(x_0)|\leq \frac{\mu-r}{\sigma^2}\frac{J_x(x)}{g(x_1)}\leq \frac{(\mu-r)K_3}{\sigma^2g(x_1)},~~~\forall x\in[x_1,x_2].$$

On the other hand, one has
\begin{equation*}
J^2_{xx}(x)+\lambda_1 J_x^2(x)-\lambda_2 J(x) \leq Z(x_0)= J^2_{xx}(x_0)+\lambda_1 J_x^2(x_0)-\lambda_2 J(x_0),  ~~\forall x\in[z_1,z_2].
\end{equation*}It implies that
$$J_{xx}^2(x)\leq J^2_{xx}(x_0)+\lambda_1J^2_x(x_0)+\lambda_2 K_3, ~~\forall x\in[x_1,x_2].$$
Thus, in view of the boundedness of $J_{xx}(x_0)$ and $J_x(x_0)$ (both not depending on $\epsilon$), finally \eqref{eq:bound sec der} holds.

{\bf Case 2 (b):} The case
\begin{equation*}
x_0 \in A:=\left\{x\in[y_1,y_2]:-{\mu-r\over \sigma^2}{J_x(x)\over J_{xx}(x)}<g(x)\right\}.
\end{equation*}
The proof of Case 2 (b) is rather lengthy, so we defer it to Appendix B.  \hfill$\Box$

\begin{proof}[Proof of Theorem \ref{th-c2}]

For every $\epsilon> 0$, we consider the function $J^{\epsilon}(x)$. By Proposition \ref{prop:smooth-epsilon}, $J^{\epsilon}(x)$ is the $C^2$ smooth solution of equation (\ref{eq:HJB-epsilon}). Therefore, by Proposition \ref{prop:smooth}, we obtain
$$-{\mu-r\over \sigma^2}{J_x^{\epsilon}(x)\over J_{xx}^{\epsilon}(x)}\geq \min\left(g(x_1),{\mu-r\over \sigma^2}{K_2\over K_1}\right)>0,  \quad \forall x\in [x_1,x_2],$$ 
where $K_1$ and $K_2$ are the two upper bound in Proposition \ref{prop:smooth}, independent of $\epsilon$. The smoothness of $J$ then follows from an argument similar to that in Theorem 5.1 of \cite{Z94}.  
\end{proof}
\section{The linear leverage }
\label{sec:linear}
In this section, we examine in detail the linear specification of the leverage bound, namely, $$g(x)=kx+L, ~~x\geq0,  ~\text{where}~~k\geq0,~L\geq0.$$
Alternatively, when $k > 0$, we may write $g(x)=k(x+\bar{L})$, $x\geq0$, where $\bar{L}:=L/k\geq0$.  We present explicit solutions in certain special cases and precisely characterize the constrained and unconstrained regions. Furthermore, we establish additional properties of the optimal value function and the optimal consumption–investment strategies, particularly in comparison to the benchmark case with a leverage constraint. Finally, we conduct a comparative analysis with respect to the risk aversion parameter and the EIS.

Note that when $L=+\infty$ or $k=+\infty$, it corresponds to the case without borrowing constraints, and the optimal strategy coincides with the results in \cite{HHJ23b}. Therefore, we restrict our attention to finite values of $L$ and $k$.
Moreover, if $\frac{\mu-r}{R\sigma^2} \le k$ holds, the leverage constraint $|\pi| \leq kX+L$ is not binding and the optimal solution to \eqref{eq:ez-problem} coincides with that of the unconstrained problem \eqref{eq:ez-problem-no-constraints}–\eqref{eq:ez-without constraints}. Therefore, in the sequel we assume that
\begin{align}\label{eq:k-range}
  \frac{\mu-r}{R\sigma^2}>k,
\end{align}
so that the borrowing constraint is binding. Condition \eqref{eq:k-range} is satisfied when the expected excess return on the stock is large, or when the volatility, risk aversion, or the parameter $k$ is sufficiently small.

\subsection{A proportional leverage bound}
In this subsection, we explicitly solve the problem with a proportional leverage bound, that is, $g(x) = kx$ for $k > 0$. 


The following proposition solves the problem in this proportional leverage situation. 

\begin{proposition}\label{pro:L=0}
Suppose $L=0$ and  $0<k<\frac{\mu-r}{R\sigma^2}$. Then, for the problem \eqref{eq:ez-problem}, the optimal value function and the optimal strategy are given, respectively, by
\begin{align}
J^0(x)&=\lambda^0\cdot {x^{1-R}\over 1-R}, ~~x>0,\label{eq:L=0 J0}\\
\pi^0(X)& =kX,\label{eq:L=0 pi0}\\
c^0(X)&=\eta^0X, \label{eq:L=0 c0}
\end{align}
where 
\begin{align*}
    \eta^0:={(S-1)\over S}\big(k(\mu-r)-{1\over 2}\sigma^2 R k^2+r-{\delta\nu\over 1-R}\big),
\end{align*}
and
\begin{align*}
  \lambda^{0}:=(\eta^{0})^{-\nu S}.  
\end{align*}
In this case, the optimal wealth process is given by the following dynamics. $$X_t=x\exp\left\{\big(r+k(\mu-r)-\eta^0-\frac{1}{2}k^2\sigma^2\big)t+k\sigma B_t\right\},~~\forall t\geq0.$$
\end{proposition}
\proof
Firstly, by the definition of $\eta$, we can verify that
$$\eta^0=\eta+\frac{1-S}{S}\left(\frac{\kappa}{R}-(k(\mu-r)-\frac{1}{2}Rk^2\sigma^2)\right).$$
Since $0<k<\frac{\mu-r}{R\sigma^2}$, it implies that $\frac{\kappa}{R}>k(\mu-r)-\frac{1}{2}Rk^2\sigma^2$, and the equality only holds when $k=\frac{\mu-r}{R\sigma^2}$. Therefore, we have that
\begin{align}\label{eq-eta0dayu0}\eta^0>\eta>0.\end{align}

On the other hand, when $L=0$, the corresponding HJB equation becomes:
\begin{align}
\label{eq:HJB-L=0}
\delta\nu J=\sup_{|\pi|\leq kx}[\pi(\mu-r)J_x+{1\over 2}\sigma^2\pi^2 J_{xx}]+\sup_{c\geq 0}[f(c,J)-cJ_x]+rxJ_x, ~~~~~~ x>0.
\end{align}
By Theorem \ref{th-unique} and Theorem \ref{th-c2}, we know that \eqref{eq:HJB-L=0} has a unique smooth solution. By directly calculations, one can easily verify that $J^0$, defined in \eqref{eq:L=0 J0}, solves the equation. The optimal strategy can also be explicitly computed by Theorem \ref{th-c2}.
\qed

The following corollary shows that, compared with the Epstein–Zin benchmark model without constraints, the investor allocates a smaller proportion of wealth to the risky asset and a larger proportion to consumption.

\begin{corollary}
\label{cor:L=0}
Suppose $L=0$ and  $0<k<\frac{\mu-r}{R\sigma^2}$. For the problem \eqref{eq:ez-problem}, then the optimal strategies satisfy $\pi^0(x)<\pi^{ez}(x)$ and $c^0(x)>c^{ez}(x)$ for all $x>0$. Moreover, ${\cal B}^* = (0, +\infty).$
\end{corollary}

\subsection{A constant leverage bound}





In this subsection, we consider a special case that $g(x) = L$, where $L$ is a positive constant. This implies that the percentage invested in the risky asset satisfies $|\pi_t| \le L$, $\forall t \ge 0.$  The objective here is to characterize both the optimal value function and the optimal consumption–investment strategy under this constant leverage constraint. 



\begin{proposition}\label{prop:unconstrained is small}
There exists $x^*>0$  such that $(0,x^*)$ is included in $\cal{U}$ and such that $\pi^*(x^*)=L$.
\end{proposition}

\begin{proof}
It is a special case of Proposition \ref{prop:unconstrained is small-k>0 L} below for $k = 0$, so we omit its proof here.
\end{proof}


We explicitly characterize the constrained and unconstrained regions in Theorem \ref{thm:two region problem}. 


\begin{theorem}\label{thm:two region problem}
There exists $x^*>0$ such that ${\cal{U}}=(0,x^*)$ and ${\cal{B}}=(x^*,+\infty)$.
\end{theorem}
\begin{proof}
First, by the change of numeraire used in Section 5.2 in \cite{HHJ23a}, without loss of generality, we assume $r=0$. Then the ordinary differential equations for $V$ in the unconstrained and constrained region are,
\begin{align}
\label{eq:V-U1}
\delta\nu {J} = {S\over 1-S}{[(1-R)^{\rho}J^{\rho}]^{1\over S}\over (J_x)^{1-S\over S}} -\kappa \frac{ (J_{x})^2}{{J_{xx}}},\quad  x\in \cal{U}, 
\end{align}
and
\begin{align}
\label{eq:V-B}
\delta\nu {J} = {S\over 1-S}{[(1-R)^{\rho}J^{\rho}]^{1\over S}\over (J_x)^{1-S\over S}}  + \mu L J_x + \frac{1}{2} \sigma^2 L^2 J_{xx}, \quad x\in \cal{B},
\end{align}
respectively. 

Define a function:
\begin{equation*}
Y(x) = \mu J_x(x) + \sigma^2 L J_{xx}(x),\quad x > 0.
\end{equation*}
Then, $Y(x) < 0$, $\forall x \in {\cal U}$, and $Y(x) > 0$ for any $x \in {\cal B}$.

In what follows, we construct the elliptic differential operators on the unconstrained and constrained regions, respectively. These constructions are presented as two lemmas, with the proofs deferred to Appendix C.

\begin{lemma}
\label{lemma:example-unconstrained}
Define an elliptic operator on the unconstrained region by
\begin{align*}
	{\cal L}^{ {\cal U}}[y] \equiv& -\mu\Big\{{S-R\over 1-S}[(1-R)J]^{{\rho\over S}-1}(J_x)^{2-{1\over S}} -{1\over \sigma^2 L}[(1-R)J]^{\rho\over S}J_x^{-{1\over S}}y\\
	&+{\mu \over \sigma^2L}[(1-R)J]^{\rho\over S}J_x^{1-{1\over S}}\Big\}
-\sigma^2L\Bigg({S-R\over 1-S}\Bigg[({\rho\over S}-1)[(1-R)J]^{{\rho\over S}-2}(1-R)J_x^{3-{1\over S}}\\
&+[(1-R)J]^{{\rho\over S}-1}(2-{1\over S})J_x^{1-{1\over S}}{y-\mu J_x\over \sigma^2 L}\Bigg] -{1\over \sigma^2 L}\Big[[{\rho\over S}(1-R)J]^{{\rho\over S}-1}(1-R)J_x^{1-{1\over S}}\\
&-[(1-R)J]^{\rho\over S}{1\over S}J_x^{-1-{1\over S}}J_{xx}\Big]y-{1\over \sigma^2 L}[(1-R)J]^{\rho\over S}J_x^{-{1\over S}}y'\\
    &+{\mu \over \sigma^2L}\Big[{\rho\over S}[(1-R)J]^{{\rho\over S}-1}(1-R)J_x^{2-{1\over S}}+[(1-R)J]^{\rho \over S}(1-{1\over S})J_x^{-{1\over S}}({y-\mu J_x\over\sigma^2 L})\Big]\Bigg)\\
    &- \frac{\kappa ({J_x})^2}{({J_{xx}})^2} y'' -  \frac{2 \kappa J_x J_{xxx}}{(J_{xx})^3} \left\{ \frac{{J_{xx}}}{\sigma^2 L} y - \frac{{J_x}}{\sigma^2 L} y' \right\}  + ( \delta\nu  + 2\kappa)y.
	\end{align*}
Then ${\cal L}^{ {\cal U} }[Y] =  0 $ in ${\cal U}$. 
\end{lemma}

\begin{lemma}
\label{lemma:example-constrained}
    Define an elliptic operator on the constrained region by
\begin{align*}
{\cal L}^{{\cal B}}[y] =& - \frac{1}{2} \sigma^2 L^2 y'' - \mu L y' +    \delta\nu y  -\mu\Big\{{S-R\over 1-S}[(1-R)J]^{{\rho\over S}-1}(J_x)^{2-{1\over S}}\\
& -{1\over \sigma^2 L}[(1-R)J]^{\rho\over S}J_x^{-{1\over S}}y+{\mu \over \sigma^2L}[(1-R)J]^{\rho\over S}J_x^{1-{1\over S}}\Big\}\\
&-\sigma^2L\Bigg({S-R\over 1-S}\Bigg[({\rho\over S}-1)[(1-R)J]^{{\rho\over S}-2}(1-R)J_x^{3-{1\over S}}\\
&+[(1-R)J]^{{\rho\over S}-1}(2-{1\over S})J_x^{1-{1\over S}}{y-\mu J_x\over \sigma^2 L}\Bigg] -{1\over \sigma^2 L}\Big[[{\rho\over S}(1-R)J]^{{\rho\over S}-1}(1-R)J_x^{1-{1\over S}}\\
&-[(1-R)J]^{\rho\over S}{1\over S}J_x^{-1-{1\over S}}J_{xx}\Big]y-{1\over \sigma^2 L}[(1-R)J]^{\rho\over S}J_x^{-{1\over S}}y'\\
    &+{\mu \over \sigma^2L}\Big[{\rho\over S}[(1-R)J]^{{\rho\over S}-1}(1-R)J_x^{2-{1\over S}}+[(1-R)J]^{\rho \over S}(1-{1\over S})J_x^{-{1\over S}}({y-\mu J_x\over\sigma^2 L})\Big]\Bigg).
\end{align*}
Then $ {\cal L}^{{\cal B}}[Y] = 0$ in ${\cal B}.$  Moreover, ${\cal L}^{ {\cal B}}[0]={\cal L}^{ {\cal U}}[0]$. 
\end{lemma}

{\em We return to the proof of Theorem \ref{thm:two region problem}:}

 By straightforward calculation, we obtain
\begin{align*}
\label{question}
{\cal L}^{ {\cal B}}[0]={\cal L}^{ {\cal U}}[0]=
& -\mu\Big\{{S-R\over 1-S}[(1-R)J]^{{\rho\over S}-1}(J_x)^{2-{1\over S}} +{\mu \over \sigma^2L}[(1-R)J]^{\rho\over S}J_x^{1-{1\over S}}\Big\}\\
&-\sigma^2L\Big\{{S-R\over 1-S}\Bigg[({\rho\over S}-1)[(1-R)J]^{{\rho\over S}-2}(1-R)J_x^{3-{1\over S}}\\
&-[(1-R)J]^{{\rho\over S}-1}(2-{1\over S})J_x^{2-{1\over S}}{\mu\over \sigma^2 L}\Bigg] +{\mu \over \sigma^2L}\Big[{\rho\over S}[(1-R)J]^{{\rho\over S}-1}(1-R)J_x^{2-{1\over S}}\\
&-[(1-R)J]^{\rho \over S}(1-{1\over S})J_x^{1-{1\over S}}({\mu \over\sigma^2 L})\Big]\Big\}.
\end{align*}
For simplicity, let
\begin{align*}
h(x) \equiv {\cal L}^{ {\cal B}}[0]={\cal L}^{ {\cal U}}[0].
\end{align*}
By multiplying $(J_x)^{-(1-{1\over S})}\cdot [(1-R)J]^{2-{\rho\over S}}$ and simplifying, we obtain that $h(x)$ has the same sign as, 
\begin{align*}
-\Big[({\sigma^2}L)^2({R-S})R J_x^2-2\mu\sigma^2L(R-S)(1-R)JJ_x+{\mu^2}[(1-R)J]^2\Big].
\end{align*}
Observe that
\begin{align*}
&-\Big[({\sigma^2}L)^2({R-S})R J_x^2-2\mu\sigma^2L(R-S)(1-R)JJ_x+{\mu^2}[(1-R)J]^2\Big]\\
<&-\Big[({\sigma^2}L)^2({R-S})^2 J_x^2-2\mu\sigma^2L(R-S)(1-R)JJ_x+{\mu^2}[(1-R)J]^2\Big]\\
=&-[\sigma^2L(R-S)J_x-\mu(1-R)J]^2\leq 0.
\end{align*}
We therefore conclude that $h(x)<0$ for all $x\in \mathbb{R}_{++}$. 

In the end, we prove Theorem \ref{thm:two region problem} by contradiction. Suppose there exists $x^{**}<+\infty$ such that $(x^*, x^{**}) \subseteq {\cal B}$ and $Y(x^{**}) = 0$. Moreover, assume that there exists $x^{***} > x^{**}$ such that $(x^{**}, x^{***}) \subseteq {\cal U}$.  We show that this leads to a contradiction, thereby completing the proof.

We first show that the constant function $y=0$ is {\em not} the supersolution for ${\cal L}^{{\cal B}}[y ]= 0$ in the region $(x^*, x^{**})$. The reason is as follows. Otherwise, since ${\cal L}^{{\cal B}}[Y] = 0$ in the region $(x^*, x^{**}) \subseteq {\cal B}$ and $Y(x^*) = Y(x^{**}) = 0$, then by the comparison principle, $Y(x) \le y=0$ for $x \in (x^*, x^{**})$. However, by its definition of ${\cal B}$, $Y(x) > 0$ for all $x \in (x^*, x^{**})$. This contradiction shows that the constant function $y=0$ is not a supersolution of ${\cal L}^{{\cal B}}[y ]= 0$ in the region $(x^*, x^{**})$.
Therefore, there exists some $x_0 \in (x^*, x^{**})$ such that, at $x=x_0$,
\begin{align*}
{\cal L}^{{\cal B}}[0] = h(x_0) < 0.
\end{align*}

We have shown that $h(x)<0$ for all $x\in (0,+\infty)$. We now consider the region $(x^{**}, x^{***}) \in {\cal U}$ and the operator ${\cal L}^{{\cal U}}$. Since ${\cal L}^{{\cal U}}[0] \le 0$ in this small region, the constant function $y=0$ is the subsolution for ${\cal L}^{{\cal U}}[0 ] = 0$.  Since  $Y(x^{**}) = Y(x^{***}) = 0$, by the comparison principle, we obtain $Y(x) \ge 0$, $\forall x \in (x^{**}, x^{***})$, which is impossible since $Y(x)$ is strictly negative over the region $(x^{**}, x^{***}) \subseteq {\cal U}$, the unconstrained region.

From the above proof, we conclude that $(x^*, +\infty) = {\cal B}$ for some $x^* > 0$,  established via a contradiction argument.
\end{proof}

\begin{proof}[Proof of Theorem \ref{th-linear}]

In the first situation, $L = 0$, it follows from Proposition \ref{pro:L=0} 
and Corollary \ref{cor:L=0}. In the second case where $k=0$, it follows from Theorem \ref{thm:two region problem}. Since we have shown that the optimal value function is the unique $C^2$ smooth solution of the HJB equation (\ref{eq:HJB}), and that the constrained region is characterized uniquely by $(0, x^*)$ for some $x^* > 0$, this threshold $x^*$ and the optimal value function are uniquely determined by the smooth-fit condition at $x^*$. Specifically, in the region $(0, x^*)$, $J(x)$ satisfies the second-order nonlinear ODE (\ref{eq:V-U1}) with  boundary condition $J(0) = 0$ and $J_{x}(0) = +\infty$; while in the region $(x^*, +\infty)$, $J(x)$ satisfies another second-order nonlinear ODE (\ref{eq:V-B}). 
\end{proof}

\begin{rem}In contrast to Proposition \ref{pro:L=0}, there is no explicit solution to the two ODEs in either the unconstrained or the constrained region. This remains true even in the particular case of $R=S$. Our contribution is to uniquely characterize the optimal value function together with $x^*$ through the two ODEs together with the smooth-fit condition. 
\end{rem}

\subsection{A general linear leverage}
In this subsection, we address the general linear leverage bound, that is, $k > 0$ and $L > 0$. While there is no explicit solution in the general case, we present further properties of the optimal value function and the consumption-investment strategy.



\begin{proposition}\label{pro:kx+l}
Suppose $0<k<\frac{\mu-r}{R\sigma^2}$ and $L=k\bar{L}>0$. Then, for the problem \eqref{eq:ez-problem}, the optimal value function and the associated optimal consumption strategy satisfy the following properties:
\begin{itemize}
    \item[(i)] $J^0(x)\leq J(x)\leq J^{ez}(x)$ and $J(x)\leq J^0(x+\bar{L})$, ~~$\forall x\geq0$;
    \item[(ii)] $\max\left\{c^{ez}(x), c^{0}(x)\left(\frac{x}{x+\bar{L}}\right)^{\frac{1}{S}-1}\right\}\leq c^*(x)\leq c^{0}(x)\left(\frac{x+\bar{L}}{x}\right)^{\frac{1}{S}}$, ~~$\forall x>0$;
    \item[(iii)] $\lim_{x\rightarrow +\infty}\frac{c^*(x)}{x}=\eta^0$;
    \item[(iv)] $\lim_{x\rightarrow 0}c^*(x)=0$.\end{itemize}
\end{proposition}
\proof  We denote the admissible strategy set by $\mathcal{A}(x;k,L)$ better highlight its dependence on the parameters.

The first part of (i) is straightforward, while the second part follows from the inclusion $\mathcal{A}(x;k,L)\subset\mathcal{A}(x+\bar{L};k,0)$.

(ii).  For any given $(\pi,c)\in \mathcal{A}(x;k,L)$, we can get $(m\pi,mc)\in\mathcal{A}(mx;k,mL)$ for any $m>0$. Moreover, by the uniqueness result in Theorem 6.9 of \cite{HHJ23b}, the utility satisfies a homogeneity property of degree $(1-R)$, namely, $V_0^{mc}=m^{1-R}V_0^{c}$ for all $m>0$. Therefore, for any $x>0$,
\begin{equation}\label{eq:homo-J}
J(x;k,L)=\sup_{(\pi,c)\in\mathcal{A}(x;k,L)}V_0^c=\sup_{(\tilde{\pi},\tilde{c})\in\mathcal{A}(1;k,L/x)}V_0^{x\tilde{c}}=x^{1-R}\sup_{(\tilde{\pi},\tilde{c})\in\mathcal{A}(1;k,L/x)}V_0^{\tilde{c}}.
\end{equation}
It implies that $J(x;k,L) = x^{1-R}J(1;k,L/x)$. Therefore, $J(x;k,L)/x^{1-R}$ is decreasing in $x$. Hence, we obtain
\begin{equation}\label{eq:J-daoshu-upperbound}
 J_x(x)x\leq (1-R)J(x),~~~ \forall x>0.
 \end{equation}

By (i) and $0 < R < 1$,  it implies that \begin{align*}
  (1-R)J^0(x)&\leq (1-R)J(x)\leq (1-R)J^{ez}(x),\\
  (1-R)J^0(x)&\leq (1-R)J(x)\leq (1-R)J^{0}(x+\bar{L}).
\end{align*}
By Theorem \ref{th-c2}, $c^*(x)=(J_x(x))^{-1/S} \big((1-R)J(x)\big)^{\frac{\rho}{S}}$,  and the fact that $(J_x(x))^{-\frac{1}{S}}x^{-\frac{1}{S}}\geq ((1-R)J(x))^{-\frac{1}{S}}$, one can derive
\begin{align*}
c^*(x)=(J_x(x))^{-1/S} \big((1-R)J(x)\big)^{\frac{\rho}{S}}
&\geq x^{\frac{1}{S}}((1-R)J(x))^{\frac{\rho-1}{S}}\\
&\geq x^{\frac{1}{S}}((1-R)J^{ez}(x))^{\frac{\rho-1}{S}}=\eta x=c^{ez}(x),
\end{align*}where we used the facts that $(1-\rho)\nu=1$ and  $J^{ez}(x)=\eta^{-\nu S}\frac{x^{1-R}}{1-R}$. Similarly, it also implies
$$c^*(x)\geq x^{\frac{1}{S}}((1-R)J^{0}(x+\bar{L}))^{\frac{\rho-1}{S}}=\eta^0 x \left(\frac{x}{x+\bar{L}}\right)^{\frac{1}{S}-1}=c^{0}(x)\left(\frac{x}{x+\bar{L}}\right)^{\frac{1}{S}-1}.$$

On the other hand, similar to \eqref{eq:homo-J}, we can also get that \begin{align*}
J(x;k,L)=\sup_{(\pi,c)\in\mathcal{A}(x;k,L)}V_0^c&=\sup_{(\tilde{\pi},\tilde{c})\in\mathcal{A}(\tilde{x};k,L/(x+\bar{L}))}V_0^{(x+\bar{L})\tilde{c}}\\
&=(x+\bar{L})^{1-R}J(\tilde{x};k,\frac{L}{x+\bar{L}})=(x+\bar{L})^{1-R}J(\tilde{x};k,k(\tilde{x}-1)),\end{align*}
where $\tilde{x}=\frac{x}{x+\bar{L}}$. One can easily find that  $\frac{J(x; k, L)}{(x+\bar{L})^{1-R}}$ is increasing in $x$. Thus, one gets that
\begin{align}\label{eq:J-daoshu-lowerbound}
   J_x(x)(x+\bar{L})\geq J(x)(1-R), ~~\forall x>0.
\end{align}Therefore, repeating the above procedures, one has
\begin{align*}
c^*(x)=(J_x(x))^{-1/S} \big((1-R)J(x)\big)^{\frac{\rho}{S}}
&\leq (x+\bar{L})^{\frac{1}{S}}((1-R)J(x))^{\frac{\rho-1}{S}}\\
&\leq(x+\bar{L})^{\frac{1}{S}}((1-R)J^{0}(x))^{\frac{\rho-1}{S}}=c^{0}(x)\left(\frac{x+\bar{L}}{x}\right)^{\frac{1}{S}}.
\end{align*}

The result for (iii) follows directly from (ii) together with the inequality $\eta^0>\eta$.

For (iv), by Theorem \ref{th-unique}, we know that
$$\delta\nu {J}\geq  {S\over 1-S}c^*(x) J_x+rx J_x,~~~\forall x>0,$$
combining $J(0)=0$ with $J_x(0)=+\infty$ together, one then can recognize that   $\lim_{x\rightarrow 0}c^*(x)=0$.
\qed

We next discuss the properties of the optimal investment strategy. We first show that the constraint is not binding when $x$ is small. 

\begin{proposition}
 \label{prop:unconstrained is small-k>0 L}
Suppose $0\leq k<\frac{\mu-r}{R\sigma^2}$ and $L>0$. Then,  there exists $\hat{x}>0$  such that $(0,\hat{x})$ is included in $\cal{U}$ and such that $\pi^*(\hat{x})=k\hat{x}+L$.
\end{proposition}
\proof We prove the result by contradiction, suppose that there exists a sequence $x_n \rightarrow 0$ such that $\{x_n\}\subset \mathcal{B}$, i.e., $\pi^*(x_n)=kx_n+L$. From the equation in the constrained domain and the definition of constrained domain, at $x_n$, we have:
\begin{align*}
\delta\nu J(x_n)\geq {(\mu-r)(kx_n+L)\over 2}J_x(x_n)\geq {L(\mu-r)\over 2}J_x(x_n).
\end{align*}
Sending $x_n\rightarrow 0$, the LHS of the above equation tends to zero. However, by Proposition \ref{pro:basic pro}, ${(\mu-r)L\over 2}J_x(x_n)\rightarrow +\infty$. We get the desired contradiction.
\qed

\begin{rem}
We conjecture that the unconstrained region ${\cal U}$ for the general linear leverage is an open interval, $(0, x^*)$
for some $x^* >0$, as shown in Theorem \ref{th-unique} for $k = 0$. However, our argument cannot be applied in the general case without imposing further conditions. Indeed, it is also not obvious for the standard CRRA utility that $R = S$ in \cite{VZ97}. 
\end{rem}

To finish the discussion of this section, we examine how the leverage constraint affect the investment strategy.


\begin{proposition}\label{pro:kx+l--portfolio}
Suppose $0<k<\frac{\mu-r}{R\sigma^2}$ and $L=k\bar{L}>0$. Then for the problem \eqref{eq:ez-problem}, for any  $x\in\mathcal{U}$ and $x>0$, the optimal strategy satisfies the following properties:
\begin{itemize}
    \item[(i)]  $\pi^{*}(x)\leq \pi^{ez}(x)+\frac{2\delta\nu\bar{L}}{(1-R)(\mu-r)}$;
   \item[(ii)] $\pi^{*}(x)\geq \pi^0(x)+\frac{2}{\mu-r}\frac{S \eta^{0}}{1-S}\left(1-\left(\frac{x+\bar{L}}{x}\right)^{1/S}\right)$;
   \item[(iii)] $\lim_{x\rightarrow 0}\pi^*(x)=0$.
    \end{itemize}
\end{proposition}
\proof (i). In the unconstrained region $\mathcal{U}$, by incorporating the optimal strategy $c^*(x)=(J_x(x))^{-1/S} \big((1-R)J(x)\big)^{\frac{\rho}{S}}$
 and  $\pi^*(x)= -\frac{\mu-r}{\sigma^{2}} \frac{J_x(x)}{J_{xx}(x)}$ into the PDE \eqref{eq:J-U1}, then it is equivalent to
\begin{align}\label{eq:JU-1-pi}
    \delta\nu \frac{J}{J'} =\frac{\mu-r}{2}\pi^{*}(x)+ {S\over 1-S}c^{*}(x)+rx,  ~~~~~x\in\mathcal{U}.
\end{align}
By the inequality \eqref{eq:J-daoshu-lowerbound} and Proposition
\ref{pro:kx+l} (ii), one can derive that
\begin{align*}
    \pi^{*}(x)&=\frac{2}{\mu-r}\left(\delta\nu \frac{J}{J'}-{S\over 1-S}c^{*}(x)-rx\right)\\
    &\leq \frac{2}{\mu-r}\left( \delta\nu\frac{x+\bar{L}}{1-R}-\frac{S}{1-S}\eta x-rx\right)=\pi^{ez}(x)+\frac{2\delta\nu\bar{L}}{(1-R)(\mu-r)},
\end{align*}where uses the definition of $\eta$. The proof for the constrained region is obvious.

(ii). Similarly, by the inequality \eqref{eq:J-daoshu-upperbound},  one can derive that
\begin{align*}
    \pi^{*}(x)&=\frac{2}{\mu-r}\left(\delta\nu \frac{J}{J'}-{S\over 1-S}c^{*}(x)-rx\right)\\
    &\geq \frac{2}{\mu-r}\left( \frac{\delta\nu}{1-R}-\frac{S}{1-S}\eta^0 -r\right)x+\frac{2}{\mu-r}\frac{S\eta^0}{1-S}\left(1-\left(\frac{x+\bar{L}}{x}\right)^{1/S}\right).
\end{align*}
Since the term
\begin{align*}
 \frac{\delta\nu}{1-R}-\frac{S}{1-S}\eta^0 -r=\frac{\kappa}{R}+\frac{S}{1-S}(\eta-\eta^0)=k(\mu-r)-\frac{1}{2}Rk^2\sigma^2,
\end{align*}taking account of $k<\frac{\mu-r}{R\sigma^2}$, then one can finally get that
$$\pi^{*}(x)\geq kx+\frac{2}{\mu-r}\frac{S\eta^0}{1-S}\left(1-\left(\frac{x+\bar{L}}{x}\right)^{1/S}\right).$$

(iii). The proof follows from equation \eqref{eq:JU-1-pi}, Propositions \ref{prop:unconstrained is small-k>0 L} and \ref{pro:kx+l} (iv), and the facts that $J_x(0)=+\infty$ and $J(0)=0$.
\qed

\section{Conclusions}
\label{sec:conclusion}

In this paper, we solve an optimal portfolio problem under Epstein–Zin preferences with a general leverage constraint in an infinite-horizon continuous-time setting. We establish a new type of existence and uniqueness theorem for the viscosity solution of the associated Hamilton–Jacobi–Bellman (HJB) equation under constraints. The HJB equation arises from Epstein–Zin utility with a non-Lipschitz aggregator. The leverage constraint, motivated by economics and finance, restricts the choice variables (risky investments) in terms of the state variables (wealth). Our approach relies on developing a dynamic programming principle tailored to this setting. 
Furthermore, we establish the smoothness property of the optimal value function based on this type of viscosity solution theorem. We also explicitly characterize the optimal consumption and investment strategy under the linear leverage constraint. In addition, we identify the constrained and unconstrained regions and compare the solution with the benchmark case in the absence of constraints. 

While the leverage constraint is one of the most widely studied constraints in the literature, it remains technically challenging because of its time-varying nature, the non-compactness of the constraint set, and the involvement of both control and state variables. Moreover, characterizing the constrained region relies on the $C^2$-smoothness property of the value function. Most previous studies in this area impose an {\em ex ante} assumption on the smoothness of the value function and then rely on a verification theorem, which typically works only when an analytical expression of the value function is available in certain special cases. This paper advances the literature by establishing new results on the existence, uniqueness, and smoothness of the viscosity solution to the HJB equation under general leverage constraints. It further generalizes portfolio choice with Epstein–Zin preferences beyond the standard time-separable models with constraints.

Although the primary contribution is theoretical, the results carry significant implications for economic and financial applications, given the widespread use of Epstein–Zin preferences and the leverage constraint in many business contexts. Moreover, the method developed in this paper is sufficiently general to accommodate other important constraints and to extend to broader classes of HJB equations under alternative preference specifications.


\renewcommand {\theequation}{A-\arabic{equation}} \setcounter
{equation}{0}
\renewcommand {\thelemma}{A.\arabic{lemma}} \setcounter
{theorem}{0}
\setcounter{equation}{0}

\section*{Appendix A: Proof of Lemma \ref{lemma:unqueness} }
\label{appex:uniqueness}
Noting $\theta$ is independent of $\lambda$, for fixed $\theta$, there exists an $M_{\theta}>0$ that is sufficiently large that $\varphi(x,y)<0$ for all $|x|\vee|y|\geq M_{\theta}$. Thus, we have
    \begin{eqnarray}\label{12057}
    |x_0|\vee|y_0|<M_{\theta}.
    \end{eqnarray}
    Next, $2\varphi(x_0,y_0)\geq \varphi(x_0,x_0)+\varphi(y_0,y_0)$
    implies
     \begin{align}\label{12058jia}
            2\lambda|x_0-y_0|^2\leq u(y_0)-u(x_0)+v(x_0)-v(y_0).
  \end{align}
    Then we get $$
    |x_0-y_0|^2\rightarrow 0\ \mbox{as}\ \lambda\rightarrow +\infty.
    $$
    Letting $\lambda\rightarrow+\infty$ in \eqref{12058jia}, by the continuity of $u$ and $v$ we get (\ref{12055}).\par
    From $\varphi(x_0,y_0)\geq \varphi(\bar{x},\bar{x})$ we have
    \begin{eqnarray}\label{12058}
                u(x_0)-v(x_0)-\theta x_0
                &\geq&
                u(\bar{x})-v(\bar{x})-\theta\bar{x}
                +v(y_0)-v(x_0)\\
                &&+\lambda|x_0-y_0|^2+\frac{\theta}{2}(y_0-x_0).
                \nonumber
    \end{eqnarray}
    By (\ref{12057}), there exists $\bar{x}_0(\theta)$ such that $\lim_{\lambda\rightarrow+\infty}x_0(\theta,\lambda)=\bar{x}_0(\theta)$.
    By sending $\lambda\rightarrow+\infty$, (\ref{12058}) combined with (\ref{12055}) implies, for $\bar{x}_0:=\bar{x}_0(\theta)$,
    \begin{eqnarray}\label{12059}
                u(\bar{x}_0)-v(\bar{x}_0)-\theta \bar{x}_0
                \geq
                u({x})-v({x})-\theta{x}, \ \mbox{for all} \ x\in \mathbb{R}_+.
    \end{eqnarray}
    We now send $\theta\rightarrow0$. If $\lim_{\theta\rightarrow0^+}[\theta \bar{x}_0(\theta)
    ]=\alpha>0$, again along subsequences, (\ref{12059}) yields
    $\sup_{\mathbb{R}_+}[u-v]-\alpha\geq \sup_{\mathbb{R}_+}[u-v]$ which contradicts that $\alpha>0$. Then we get
   \begin{eqnarray}\label{03181zhou}
    \lim_{\theta\rightarrow0^+}[\theta \bar{x}_0(\theta)
    ]=0.
   \end{eqnarray}
    From $g$ is Lipschitz, it follows that (\ref{101025}) holds true. Moreover, by (\ref{12059}), we have
\begin{eqnarray*}
               u(\bar{x}_0)-v(\bar{x}_0)\geq
               \sup_{x\in \mathbb{R}_+}[u(x)-v(x)-\theta x]\geq\tilde{m}>0
               \geq u(0)-v(0),
\end{eqnarray*}
then $\bar{x}_0>0$.
\begin{proof}[ Proof of (\ref{101125})]
Case 1.  $\lim_{\lambda\rightarrow+\infty}|x_0-y_0||Y|=0$.

By the properties of $g$, we have 
\begin{eqnarray*}
|I|\leq K(\mu-r)|x_0-y_0|\left(\lambda|x_0-y_0|+\frac{\theta}{2}\right)+{K\over 2}\sigma^2 (g(x_0)+g(y_0))|x_0-y_0|[|Y|].
\end{eqnarray*}
Then, by (\ref{12055}) and (\ref{12057}), we get (\ref{101125}).

Case 2. $\limsup_{\lambda\rightarrow+\infty}|x_0-y_0||Y|>0$. Without loss of generality, we may assume that 
\begin{eqnarray}\label{1011251}
\lim_{\lambda\rightarrow+\infty}|x_0-y_0||Y|=\bar{a}>0.
\end{eqnarray}
From the strict concavity of $v$,
$$
            0<Y=\lim_{k\rightarrow\infty}\partial_{xx}\hat{\psi}_k(y_k).
$$
Then, the function 
$$\bar{f}(\pi)=-\pi(\mu-r)\left[\lambda(y_0-x_0)+\frac{\theta}{2}\right]-{1\over 2}\sigma^2\pi^2 Y$$
 attains its maximum at $\hat{\pi}=\frac{(\mu-r)\left[\lambda(y_0-x_0)+\frac{\theta}{2}\right]}{\sigma^2Y}=\frac{(\mu-r)\left[\lambda(y_0-x_0)+\frac{\theta}{2}\right](y_0-x_0)}{\sigma^2Y(y_0-x_0)}$. 
 From (\ref{12055}) and (\ref{1011251}),
 $$
 \lim_{\lambda\rightarrow+\infty}\hat{\pi}=0.
 $$
 Therefore, by (\ref{x_0y_0>0}) and the properties of $g$, there exists a constant $\hat{\Delta}_\theta>0$ large enough that
 $$|\hat{\pi}|\leq g(x_0)\wedge  g(y_0)\quad \mbox{and}\ \ I=0,\qquad \mbox{for all }
\ \ \lambda\geq \hat{\Delta}_\theta.
$$
 Thus, we obtain (\ref{101125}).
\end{proof}

\renewcommand {\theequation}{B-\arabic{equation}} \setcounter
{equation}{0}
\renewcommand {\thelemma}{B.\arabic{lemma}} \setcounter
{theorem}{0}
\setcounter{equation}{0}
\section*{Appendix B: Proof of Case 2 (b) in Proposition \ref{prop:smooth} }
\label{appen:smooth}

{\bf Case 2 (b):} The case
\begin{equation*}
x_0 \in A:=\left\{x\in[y_1,y_2]:-{\mu-r\over \sigma^2}{J_x(x)\over J_{xx}(x)}<g(x)\right\}.
\end{equation*}

In this unconstrained domain $A$, the HJB equation has the form:
\begin{equation}
\label{eq: HJB-unconstrained}
\delta\nu J=-\kappa {J_x^2\over J_{xx}}+{1\over 2}\epsilon^2\sigma^2x^2J_{xx}+ {S\over 1-S}((1-R)J)^{\rho\over S}(J_x)^{1-{1\over S}}+rxJ_x.
\end{equation}
Since $A$ is an open set,  then $J$ is smooth enough (we will use the third and fourth derivatives in the sequel). Our goal is to obtain an estimate of $J_{xx}(x_0)$.

Since $Z$ attains the (locally) maximum at $x_0\in A$,  then we have
\begin{equation*}
Z_x(x_0)=0, ~~Z_{xx}(x_0)\leq 0,
\end{equation*}
where the derivatives are
\begin{equation}
\label{eq:first derivative}
Z_x=2\xi\xi_xJ^2_{xx}+2\xi^2J_{xx}J_{xxx}+2\lambda_1J_xJ_{xx}-\lambda_2J_x
\end{equation}
and
\begin{eqnarray}
\label{eq: second derivative}
Z_{xx}&=&2\xi_x^2J^2_{xx}+2\xi\xi_{xx}J^2_{xx}+8\xi\xi_{x}J_{xx}J_{xxx}+2\xi^2J^2_{xxx}+2\xi^2J_{xx}J_{xxxx}\\  \nonumber
&&+2\lambda_1J^2_{xx}+2\lambda_1J_{x}J_{xxx}-\lambda_2J_{xx}.
\end{eqnarray}

Differentiate (\ref{eq: HJB-unconstrained}) once, we get
\begin{align}
\label{eq:HJB-unconstrained-once}
&-\kappa J_x-{1\over 2}\epsilon^2\sigma^2x^2J_{xxx}-\kappa {J_x^2J_{xxx}\over J^2_{xx}}=J_{xx}\Big[rx+\epsilon^2\sigma^x-[(1-R)J]^{\rho\over S}(J_x)^{-{1\over S}}\Big]\\
&~~~\qquad\qquad+ {S-R\over 1-S}[(1-R)J]^{{\rho\over S}-1}(J_x)^{1-{1\over S}}+J_x(r-3\kappa)-\delta\nu J_x. \nonumber
\end{align}
Differentiate (\ref{eq: HJB-unconstrained}) twice, we get
\begin{align}
\label{eq: HJB-unconstrained-twice}
&-\kappa {J_x^2J_{xxxx}\over J^2_{xx}}-{1\over 2}\epsilon^2\sigma^2x^2J_{xxxx}-\kappa{J_x^2 V^2_{xxx}\over J^3_{xx}}=-\delta\nu J_{xx}+J_{xx}(2r-2\kappa+\epsilon^2\sigma^2)\\\nonumber
&+2\kappa{J_xJ_{xxx}\over J_{xx}}-3\kappa{J_x^2J^2_{xxx}\over J^3_{xx}}+J_{xxx}(rx+2\epsilon^2\sigma^2 x)+ [(1-R)J]^{{\rho\over S}-2}\Bigg\{{S-R\over 1-S}\Big[{R(S-1)\over S}(J_x)^2\\\nonumber
&+(1-R)JJ_{xx}\Big](J_x)^{1-{1\over S}}[(1-R)J]^2\Big[{1\over S}(J_x)^{-{1\over S}-1}(J_{xx})^2-(J_x)^{-{1\over S}}J_{xxx}\Big]\\\nonumber
&-2{S-R\over S}(J_x)^{2-{1\over S}}J_{xx}[(1-R)J]\Bigg\}.\nonumber
\end{align}

Since $Z_{xx}(x_0)\leq 0$, we have
\begin{equation*}
-\Big[\kappa{J_x^2(x_0)\over J_{xx}^2(x_0)}+{1\over 2}\epsilon^2\sigma^2x_0^2\Big]Z_{xx}(x_0)\geq 0.
\end{equation*}
Replace $Z_{xx}(x_0)$ in the above equation by the representation in (\ref{eq: second derivative}), at $x_0$, we have
\begin{align*}
&-\Big[\kappa{J_x^2\over J_{xx}^2}+{1\over 2}\epsilon^2\sigma^2x^2\Big]Z_{xx}=2\xi^2J_{xx}\Big[-\kappa{J_x^2J_{xxxx}\over J^2_{xx}}-{1\over 2}\epsilon^2\sigma^2x^2J_{xxxx}-\kappa {J_x^2J^2_{xxx}\over J^3_{xx}}\Big]\\
&+2\lambda_1 J_x[-\kappa J_x-{1\over2} \epsilon^2\sigma^2x^2J_{xxx}-\kappa{J_x^2 J_{xxx}\over J^2_{xx}}]+\lambda_2[\kappa {J_x^2\over J_{xx}}+{1\over 2}\epsilon^2\sigma^2 x^2 J_{xx}]-2\kappa \xi_x^2J_x^2\\
&-\epsilon^2\sigma^2x^2\xi_x^2J^2_{xx}-2\kappa J_x^2\xi \xi_{xx}-x^2\epsilon^2\sigma^2 \xi \xi_{xx}J^2_{xx}-8\kappa \xi\xi_x{J_x^2J_{xxx}\over J_{xx}}-4x^2\epsilon^2\sigma^2\xi\xi_xJ_{xx}J_{xxx}\\
&-\epsilon^2\sigma^2x^2\xi^2J^2_{xxx}-\lambda_1x^2\epsilon^2\sigma^2J^2_{xx}\geq 0.
\end{align*}
Now, using (\ref{eq: HJB-unconstrained}),(\ref{eq:HJB-unconstrained-once}) and (\ref{eq: HJB-unconstrained-twice}), plugging into the above equation, we get
\begin{align*}
&-2\delta\nu\xi^2J_{xx}^2+2\xi^2(2r-2\kappa+\epsilon^2\sigma^2)J^2_{xx}+4\kappa\xi^2J_xJ_{xxx}-6\kappa \xi^2{J_x^2J^2_{xxx}\over J^2_{xx}}+2\xi^2J_{xx}J_{xxx}\Big(rx\\\nonumber
&+2\epsilon^2\sigma^2 x-[(1-R)J]^{\rho\over S}(J_x)^{-{1\over S}}\Big)+2\xi^2J_{xx}\Bigg[[(1-R)J]^{{\rho\over S}-2}\Bigg\{{S-R\over 1-S}\Big[{R(S-1)\over S}(J_x)^2\\
&+(1-R)JJ_{xx}\Big](J_x)^{1-{1\over S}}+[(1-R)J]^2{1\over S}(J_x)^{-{1\over S}-1}(J_{xx})^2-2{S-R\over S}(J_x)^{2-{1\over S}}J_{xx}[(1-R)J]\Bigg\}\Bigg]\\
&+2\lambda_1J_xJ_{xx}\Big[rx+\epsilon^2\sigma^2x-[(1-R)J]^{\rho\over S}(J_x)^{-{1\over S}}\Big]+2\lambda_1 J_x {S-R\over 1-S}[(1-R)J]^{{\rho\over S}-1}(J_x)^{1-{1\over S}}\\
& +2\lambda_1 J^2_x(r-3\kappa)-2\lambda_1\delta\nu J^2_x+2\kappa\lambda_2{J_x^2\over J_{xx}}+\lambda_2 J_x\Big(-rx+[(1-R)J]^{\rho\over S}J_x^{-{1\over S}}\Big)\\
&-\lambda_2{1\over 1-S}[(1-R)J]^{\rho\over S}J_x^{1-{1\over S}}+\lambda_2\delta \nu J-2\kappa \xi_x^2J_x^2-\epsilon^2\sigma^2x^2\xi_x^2J^2_{xx}-2\kappa J_x^2\xi \xi_{xx}\\
&-x^2\epsilon^2\sigma^2 \xi \xi_{xx}J^2_{xx}-8\kappa \xi\xi_x{J_x^2J_{xxx}\over J_{xx}}-4x^2\epsilon^2\sigma^2\xi\xi_xJ_{xx}J_{xxx}-\epsilon^2\sigma^2x^2\xi^2J^2_{xxx}-\lambda_1x^2\epsilon^2\sigma^2J^2_{xx}\geq 0.
\end{align*}
Rearranging the terms leads to
\begin{align}
\label{eq:HJB unconstrained-middle step}
&\Big[rx+\epsilon^2\sigma^2x-[(1-R)J]^{\rho\over S}(J_x)^{-{1\over S}}\Big][2\xi^2J_{xx}J_{xxx}+2\lambda_1J_xJ_{xx}-\lambda_2J_x]+2\xi^2\epsilon^2\sigma^2 x J_{xx}J_{xxx}\\
&+\epsilon^2\sigma^2\lambda_2xJ_x-2\delta\nu \xi^2J^2_{xx}-4(\kappa-r-{\epsilon^2\sigma^2\over 2})\xi^2J_{xx}^2+4\kappa\xi^2J_xJ_{xxx}-6\kappa\xi^2{J_x^2J^2_{xxx}\over J^2_{xx}}\nonumber\\
&+2\xi^2J_{xx}\Bigg[[(1-R)J]^{{\rho\over S}-2}\Bigg\{{S-R\over 1-S}\Big[{R(S-1)\over S}(J_x)^2+(1-R)JJ_{xx}\Big](J_x)^{1-{1\over S}}\nonumber\\ 
&-2{S-R\over S}(J_x)^{2-{1\over S}}J_{xx}[(1-R)J]\Bigg\}\Bigg]+2\lambda_1 J_x {S-R\over 1-S}[(1-R)J]^{{\rho\over S}-1}(J_x)^{1-{1\over S}}-2\lambda_1\delta\nu J_x^2\nonumber\\
&-2\lambda_1(3\kappa-r)J_x^2+\lambda_2\delta\nu J+2\kappa\lambda_2{J_x^2\over J_{xx}}-2\kappa\xi_x^2J_x^2-2\kappa J_x^2\xi\xi_{xx}-8\kappa \xi\xi_x{J_x^2J_{xxx}\over J_{xx}}\nonumber\\
&+\epsilon^2\sigma^2x^2[4\xi\xi_x|J_{xx}|J_{xxx}-\xi^2J^2_{xxx}-\lambda_1 J^2_{xx}]-\epsilon^2\sigma^2x^2\xi_{x}^2J^2_{xx}-\epsilon^2\sigma^2x^2\xi\xi_{xx}J^2_{xx}\nonumber\\
&-\lambda_2{1\over 1-S}[(1-R)J]^{\rho\over S}J_x^{1-{1\over S}}\geq -2\xi^2 J^3_{xx}[(1-R)J]^{{\rho\over S}}{1\over S}(J_x)^{-{1\over S}-1}\nonumber.
\end{align}

Since $Z_x(x_0)=0$, by (\ref{eq:first derivative}), at $x_0$,
we may replace
$2\xi^2J_{xx}J_{xxx}+2\lambda_1J_xJ_{xx}-\lambda_2J_x$
by the term $-2\xi\xi_xJ^2_{xx}$. Therefore, we may simplify (\ref{eq:HJB unconstrained-middle step}) as
\begin{align}
\label{eq:HJB unconstrained-last step}
&-2\xi\xi_xJ^2_{xx}\Big[rx+\epsilon^2\sigma^2x-[(1-R)J]^{\rho\over S}(J_x)^{-{1\over S}}\Big]-2\delta\nu \xi^2J^2_{xx}-2\lambda_1\delta\nu J_x^2+\lambda_2\delta\nu J\\
&-2\lambda_1(3\kappa-r)J_x^2-2\kappa\lambda_2{ J_x^2 \over |J_{xx}|}-\lambda_2{1\over 1-S}[(1-R)J]^{\rho\over S}J_x^{1-{1\over S}}\nonumber\\
&-2\kappa J_x^2\xi\xi_{xx}-\epsilon^2\sigma^2x^2\xi_{x}^2J^2_{xx}-\epsilon^2\sigma^2x^2\xi\xi_{xx}J^2_{xx}-2\kappa\xi_x^2J_x^2\nonumber\\
&+\Big[-4(\kappa-r-{\epsilon^2\sigma^2\over 2})\xi^2J_{xx}^2+4\kappa\xi^2J_xJ_{xxx}-6\kappa\xi^2{J_x^2J^2_{xxx}\over J^2_{xx}}+8\kappa \xi\xi_x{J_x^2J_{xxx}\over |J_{xx}|}\Big]\nonumber\\
&+\epsilon^2\sigma^2x^2[(4\xi\xi_x-{2\xi^2\over x})|J_{xx}|J_{xxx}-\xi^2J^2_{xxx}-\lambda_1 J^2_{xx}]+\epsilon^2\sigma^2\lambda_2 xJ_x\nonumber\\
&+2\xi^2J_{xx}\Bigg[[(1-R)J]^{{\rho\over S}-2}\Bigg\{{S-R\over 1-S}\Big[{R(S-1)\over S}(J_x)^2+(1-R)JJ_{xx}\Big](J_x)^{1-{1\over S}}\nonumber\\
&-2{S-R\over S}(J_x)^{2-{1\over S}}J_{xx}[(1-R)J]\Bigg\}\Bigg]\geq -2\xi^2 J^3_{xx}[(1-R)J]^{{\rho\over S}}{1\over S}(J_x)^{-{1\over S}-1}\nonumber.
\end{align}
By the property of quadratic function, when $\epsilon$ small enough, at $x_0$ we have
\begin{align}
\label{eq:quadratic 1}
&-4(\kappa-r-{\epsilon^2\sigma^2\over 2})\xi^2J_{xx}^2+4\kappa\xi^2J_xJ_{xxx}-6\kappa\xi^2{J_x^2J^2_{xxx}\over J^2_{xx}}+8\kappa \xi\xi_x{J_x^2J_{xxx}\over |J_{xx}|}
\leq C({\xi_x^2}J_x^2+\xi^2J^2_{xx})
\end{align}
for some positive constant $C$. Similarly, at $x_0$, choose $\lambda_1$ sufficiently large, we have
\begin{align}
\label{eq: quadratic 2}
    &(4\xi\xi_x-{2\xi^2\over x})|J_{xx}|J_{xxx}-\xi^2J^2_{xxx}-\lambda_1 J^2_{xx}\leq J^2_{xx}\Big[(2\xi_x-{\xi\over x})^2-\lambda_1\Big]\leq 0.
\end{align}

Now, plugging (\ref{eq:quadratic 1}), (\ref{eq: quadratic 2}) into (\ref{eq:HJB unconstrained-last step}) and leave out all the negative terms (note that we are considering the case $S<R<1$), at $x_0$, we have
\begin{align*}
& J^2_{xx}\Big\{-2\xi\xi_{x}[rx+\epsilon^2\sigma^2x-[(1-R)J]^{\rho\over S}(J_x)^{-{1\over S}}]-\epsilon^2\sigma^2x^2\xi\xi_{xx}+C\xi^2\Big\}+J_x^2\Big[{\xi^2_x}-2\kappa\xi\xi_{xx}\\
&-2\lambda_1(3\kappa-r)\Big] +2\xi^2J_{xx}\Bigg[-2{S-R\over S}(J_x)^{2-{1\over S}}J_{xx}[(1-R)J]\Bigg\}\Bigg]+\lambda_2(\delta \nu J+\epsilon^2 \sigma^2 x)\\
&  \geq -2\xi^2 J^3_{xx}[(1-R)J]^{{\rho\over S}}{1\over S}(J_x)^{-{1\over S}-1}.
\end{align*}

Recalling from (\ref{eq: bound for V and V_x}), $J$ and $J_x$ are uniformly bounded, therefore, at $x_0$, there exist positive constants $k_1,k_2,k_3,k_4$ that
\begin{equation*}
J^2_{xx}\Big[k_1 \xi|\xi_x|+k_2\xi|\xi_{xx}|+C\xi^2\Big]+k_3\geq k_4 \xi^2|J^3_{xx}|.
\end{equation*} Since $x_0\notin \text{supp}~\xi$, this implies that $\xi(x_0)>0$.
Applying the properties (iii) of $\xi$ and $\xi\leq 1$, one can find positive constants $k_5$ and $k_6$  such that
$$k_5J^2_{xx}(x_0)+k_6\geq |J^3_{xx}(x_0)|=(J^2_{xx}(x_0))^{3/2}.$$
Therefore, it implies that
$$|J_{xx}(x_0)|\leq k_7,$$ where $k_7$ is independent of $\epsilon$. Finally,  we can get the desired result (\ref{eq:bound sec der}) by repeating the last part of the proof in Case 2 (a).

\renewcommand {\theequation}{C-\arabic{equation}} \setcounter
{equation}{0}
\renewcommand {\thelemma}{C.\arabic{lemma}} \setcounter
{theorem}{0}
\setcounter{equation}{0}
\section*{Appendix C: Proof of Lemmas \ref{lemma:example-unconstrained} - \ref{lemma:example-constrained}}
\label{appen:regions}
{\bf Proof of Lemma \ref{lemma:example-unconstrained}:} 

In the unconstrained region, the value function ${V}(\cdot)$ satisfies the ODE (\ref{eq:V-U1}).
By differentiating the ODE equation once and twice, we obtain
\begin{align*}
\delta\nu J_x =& {S-R\over 1-S}[(1-R)J]^{{\rho\over S}-1}(J_x)^{2-{1\over S}} -{1\over \sigma^2 L}[(1-R)J]^{\rho\over S}J_x^{-{1\over S}}Y+{\mu \over \sigma^2L}[(1-R)J]^{\rho\over S}J_x^{1-{1\over S}}\\
&-2\kappa J_x + \frac{\kappa (J_x)^2 J_{xxx}}{(J_{xx})^2}
\end{align*}
and
\begin{align*}
	\delta \nu J_{xx} =&  {S-R\over 1-S}\Bigg[({\rho\over S}-1)[(1-R)J]^{{\rho\over S}-2}(1-R)J_x^{3-{1\over S}}+[(1-R)J]^{{\rho\over S}-1}(2-{1\over S})J_x^{1-{1\over S}}{Y-\mu J_x\over \sigma^2 L}\Bigg] \\
    &-{1\over \sigma^2 L}\Big[[{\rho\over S}(1-R)J]^{{\rho\over S}-1}(1-R)J_x^{1-{1\over S}}-[(1-R)J]^{\rho\over S}{1\over S}J_x^{-1-{1\over S}}J_{xx}\Big]Y\\
    &-{1\over \sigma^2 L}[(1-R)J]^{\rho\over S}J_x^{-{1\over S}}Y'+{\mu \over \sigma^2L}\Big[{\rho\over S}[(1-R)J]^{{\rho\over S}-1}(1-R)J_x^{2-{1\over S}}\\
    &+[(1-R)J]^{\rho \over S}(1-{1\over S})J_x^{-{1\over S}}({Y-\mu J_x\over\sigma^2 L})\Big]\\
	& - 2 \kappa {J}_{xx} + \frac{\kappa (J_x)^2 J_{xxxx}}{(J_{xx})^2} + \frac{2 \kappa J_{x} J_{xxx}}{({J}_{xx})^3}  [({J}_{xx})^2 -J_x J_{xxx}].
\end{align*}
By the definition of $Y(W)$, the last two equations imply
\begin{align*}
\delta\nu Y &= \mu\Big\{{S-R\over 1-S}[(1-R)J]^{{\rho\over S}-1}(J_x)^{2-{1\over S}} -{1\over \sigma^2 L}[(1-R)J]^{\rho\over S}J_x^{-{1\over S}}Y+{\mu \over \sigma^2L}[(1-R)J]^{\rho\over S}J_x^{1-{1\over S}}\Big\}\\
&+\sigma^2L\Bigg({S-R\over 1-S}\Bigg[({\rho\over S}-1)[(1-R)J]^{{\rho\over S}-2}(1-R)J_x^{3-{1\over S}}\\
&+[(1-R)J]^{{\rho\over S}-1}(2-{1\over S})J_x^{1-{1\over S}}{Y-\mu J_x\over \sigma^2 L}\Bigg] -{1\over \sigma^2 L}\Big[[{\rho\over S}(1-R)J]^{{\rho\over S}-1}(1-R)J_x^{1-{1\over S}}\\
&-[(1-R)J]^{\rho\over S}{1\over S}J_x^{-1-{1\over S}}J_{xx}\Big]Y-{1\over \sigma^2 L}[(1-R)J]^{\rho\over S}J_x^{-{1\over S}}Y'\\
    &+{\mu \over \sigma^2L}\Big[{\rho\over S}[(1-R)J]^{{\rho\over S}-1}(1-R)J_x^{2-{1\over S}}+[(1-R)J]^{\rho \over S}(1-{1\over S})J_x^{-{1\over S}}({Y-\mu J_x\over\sigma^2 L})\Big]\Bigg)\\
&- 2 \kappa Y+ \frac{\kappa (J_x)^2}{J_{xx}^2} Y'' + \frac{2 \kappa J_x {J_{xxx}}}{({J_{xx}})^3} \left\{ \frac{{J_{xx}}}{\sigma^2 L} Y - \frac{{J}_x}{\sigma^2 L} Y' \right\}.
\end{align*}

\begin{proof}[Proof of Lemma \ref{lemma:example-constrained}]

In the constrained region ${\cal B}$, by differentiating the ODE (\ref{eq:V-B}) of $V(W)$ once and twice, we have
\begin{align*}
\delta\nu J_{x} =  & {S-R\over 1-S}[(1-R)J]^{{\rho\over S}-1}(J_x)^{2-{1\over S}} -{1\over \sigma^2 L}[(1-R)J]^{\rho\over S}J_x^{-{1\over S}}Y+{\mu \over \sigma^2L}[(1-R)J]^{\rho\over S}J_x^{1-{1\over S}}\\
&+ \mu L {J_{xx}} + \frac{1}{2} \sigma^2 L^2 {J_{xxx}}
\end{align*}
and
\begin{align*}
\delta\nu {J}_{xx} = &  {S-R\over 1-S}\Bigg[({\rho\over S}-1)[(1-R)J]^{{\rho\over S}-2}(1-R)J_x^{3-{1\over S}}+[(1-R)J]^{{\rho\over S}-1}(2-{1\over S})J_x^{1-{1\over S}}{Y-\mu J_x\over \sigma^2 L}\Bigg] \\
    &-{1\over \sigma^2 L}\Big[[{\rho\over S}(1-R)J]^{{\rho\over S}-1}(1-R)J_x^{1-{1\over S}}-[(1-R)J]^{\rho\over S}{1\over S}J_x^{-1-{1\over S}}J_{xx}\Big]Y\\
    &-{1\over \sigma^2 L}[(1-R)J]^{\rho\over S}J_x^{-{1\over S}}Y'+{\mu \over \sigma^2L}\Big[{\rho\over S}[(1-R)J]^{{\rho\over S}-1}(1-R)J_x^{2-{1\over S}}\\
    &+[(1-R)J]^{\rho \over S}(1-{1\over S})J_x^{-{1\over S}}({Y-\mu J_x\over\sigma^2 L})\Big]+ \mu LJ_{xxx} + \frac{1}{2} \sigma^2 L^2 J_{xxxx}.
\end{align*}
Then,
\begin{align*}
\delta\nu Y &=\mu\Big\{{S-R\over 1-S}[(1-R)J]^{{\rho\over S}-1}(J_x)^{2-{1\over S}} -{1\over \sigma^2 L}[(1-R)J]^{\rho\over S}J_x^{-{1\over S}}Y+{\mu \over \sigma^2L}[(1-R)J]^{\rho\over S}J_x^{1-{1\over S}}\Big\}\\
&+\sigma^2L\Bigg({S-R\over 1-S}\Bigg[({\rho\over S}-1)[(1-R)J]^{{\rho\over S}-2}(1-R)J_x^{3-{1\over S}}\\
&+[(1-R)J]^{{\rho\over S}-1}(2-{1\over S})J_x^{1-{1\over S}}{Y-\mu J_x\over \sigma^2 L}\Bigg] -{1\over \sigma^2 L}\Big[[{\rho\over S}(1-R)J]^{{\rho\over S}-1}(1-R)J_x^{1-{1\over S}}\\
&-[(1-R)J]^{\rho\over S}{1\over S}J_x^{-1-{1\over S}}J_{xx}\Big]Y-{1\over \sigma^2 L}[(1-R)J]^{\rho\over S}J_x^{-{1\over S}}Y'\\
    &+{\mu \over \sigma^2L}\Big[{\rho\over S}[(1-R)J]^{{\rho\over S}-1}(1-R)J_x^{2-{1\over S}}+[(1-R)J]^{\rho \over S}(1-{1\over S})J_x^{-{1\over S}}({Y-\mu J_x\over\sigma^2 L})\Big]\Bigg)\\
&+\mu L Y' + \frac{1}{2} \sigma^2 L^2 Y''.
\end{align*}
\end{proof}

\section*{Appendix D: Existence of viscosity solution for $R>1$}
\label{appen:existence}
\renewcommand {\theequation}{D-\arabic{equation}} \setcounter
{equation}{0}
\renewcommand {\thelemma}{D.\arabic{lemma}} \setcounter
{theorem}{0}
\setcounter{equation}{0}

When $R>1$,  for any $x>0$ and $c\in\mathcal{C}(x)$, we cannot guarantee $$V_0^c=\mathbb{E}\left[\int_{0}^{\infty}e^{-\delta s} f(c_{s},V_{s}^{c})ds\right]>-\infty$$ anymore.  However,  by Proposition \ref{pro:basic pro},  we can only consider the following admissible consumptions set
$$\tilde{\mathcal{C}}(x)=\left\{~c\in\mathcal{C}(x)~|~ V_0^c=\mathbb{E}\left[\int_{0}^{\infty}e^{-\delta s} f(c_{s},V_{s}^{c})ds\right]\geq \delta^{-\nu} \frac{(rx)^{1-R}}{1-R}\right\}.$$
Therefore, Lemma \ref{lem:bsde1} also holds for any $c\in\tilde{\mathcal{C}}(x)$. We can also get the following result. 

\begin{lemma}\label{lem:bsde-rs>1}
Let $1<R<S$ and finite time horizon $T>0$. For any $\zeta\in \mathcal{F}_{T}$ such that $\zeta<0$ and $\mathbb{E}[((1-R)\zeta)^{\frac{1}{\nu}}+\int_{0}^{T}e^{-\delta s}c_{s}^{1-S}ds]<+\infty$, then BSDE
\begin{align}\label{eq:bsde-fix-terminal--rs>1}
 V_t=\zeta+\int_{t}^{T}e^{-\delta s} f(c_{s},V_{s})ds
 -\int^T_tZ_sdW_s,\ \ 0\leq t\leq T
\end{align}admits a unique solution $(V, Z)$ such that $V$ is strictly negative, and of class (D), with $\int_{0}^{T}Z_{s}^{2}ds<+\infty$.
\end{lemma}
\proof
In this situation, one can construct the generator  as follows:  for any $m>0$,
 $$f_{m}(c,v)=\frac{(c\vee \frac{1}{m})^{1-S}}{1-S}\left(((1-R)v)\vee \frac{1}{m})\right)^{\rho},$$and $f_{m}$ is globally Lipschitz continuous in $v$, moreover,
 $\frac{m^{S-1-\rho}}{1-S}\leq f_{m}\leq 0$.
 Then the truncated BSDE
 \begin{align}\label{eq:bsde-fix-terminal-m-s>1}
 V_t^{c,m}=\zeta\vee (-m)+\int_{t}^{T}e^{-\delta s} f_{m}(c_{s},V_{s}^{c,m})ds
 -\int^T_tZ^{c,m}_sdW_s,\ \ 0\leq t\leq T
\end{align}admits a unique solution $(V^{c,m}, Z^{c,m})$. Besides, $f_{m}$ and the terminal random value are both decreasing with respect to $m$,  by comparison theorem,  it implies $V^{c,m}$ is also decreasing with $m$, and nonpositive (taking generator and terminal random value as 0).

Similar to Lemma \ref{lem:bsde-rs<1}, it suffices to find a lower bound for $V^{c,m}$ that is independent of $m$.
Taking the conditional expectation, one has that
\begin{align*}
&((1-R)V_{t}^{c,m})^{1-\rho}\\
\leq&\mathbb{E}\left[ ((1-R)(\zeta \vee (-m)))^{1-\rho}~\big|~\mathcal{F}_{t}\right]\\
&~~+\mathbb{E}\left[ \int_{t}^{T}  \nu(1-\rho)((1-R)V_{s}^{c,m})^{-\rho} e^{-\delta s}(c_{s}\vee \frac{1}{m})^{1-S} (((1-R)V_{s}^{c,m})\vee \frac{1}{m})^{\rho}ds
~\big|~\mathcal{F}_{t}\right]\\
\leq &\mathbb{E}\left[ ((1-R)\zeta)^{1-\rho}+\nu(1-\rho) \int_{t}^{T}e^{-\delta s}c_{s}^{1-S}ds ~\big|~\mathcal{F}_{t}\right].
\end{align*}
Thus, it implies that
\begin{align*}
V_{t}^{c,m}&\geq \frac{1}{1-R}\left(\mathbb{E}\left[ ((1-R)\zeta)^{1-\rho}+\nu(1-\rho) \int_{t}^{T}e^{-\delta s}c_{s}^{1-S}ds ~\big|~\mathcal{F}_{t}\right]\right)^{\nu}=:{\underline{V}}_{t}^{c}.
\end{align*}\hfill$\Box$

For any $N>0$, define
 $${\mathcal{A}}^N(x):=\{(\pi,c)\in \mathcal{A}(x)~|~c\in\tilde{\mathcal{C}}(x), 
 c_t\geq 1/N, ~\forall t \ge 0\},\ \ N>0, $$ 
 and
\begin{eqnarray*}
\mathbf{H}_N(x,k,p,q)&=&\sup_{{|\pi|\leq g(x)}}[\pi(\mu-r)p+\frac{1}{2}\sigma^2\pi^2q]+\sup_{c\geq 1/N}[f(c,k)-cp]+rxp,\\
&&\ \ \ \ \ \ \ \ \ \  \ \ \ \ \ \ \ \ \  \ \ \ \  \ \ \  \ \ \ \ \ \ \ \ \ \ \ \ \ (x,k,p,q)\in
\mathbb{R}_+\times \mathbb{R}_+\times \mathbb{R}\times \mathbb{R}.
\end{eqnarray*}
By the similar proof procedure of Proposition \ref{pro: dpp} and Theorem \ref{existenceN}, 
$J^N$ defined in (\ref{valueN}) satisfies DPP (\ref{dppzhou}) and is a viscosity solution of the approximation  HJB equation (\ref{eq:HJBN}). By the same proof procedure of Theorem \ref{existence}, we can show that the optimal value
                          function $J$ defined by (\ref{eq:ez-problem}) is a
                          viscosity solution to  HJB equation (\ref{eq:HJB}).

\end{document}